\documentclass[runningheads,envcountsame]{llncs}

\usepackage{amsmath}
\usepackage{amssymb}

\usepackage{hyperref}

\usepackage{float}
\usepackage{wrapfig}
\usepackage{subcaption}
\usepackage{microtype}

\usepackage{xcolor}
\usepackage{tikz}
\usetikzlibrary{positioning}
\usetikzlibrary{shapes.multipart}

\usepackage{orcidlink}
\usepackage{bbding}

\usepackage{xspace}

\spnewtheorem{claim}[theorem]{Claim}{\bfseries}{\itshape}

\renewcommand{\to}[1][]{\mathop{\xrightarrow{~#1~}}}
\newcommand{\Nat}{\mathbb{N}}

\newcommand{\of}[1]{(#1)}
\newcommand{\tuple}[1]{\langle#1\rangle}
\newcommand{\set}[1]{\{#1\}}

\newcommand{\Sizeof}[1]{\left|#1\right|}

\newcommand{\Set}[1]{\left\{#1\right\}}

\newcommand{\newclass}[2]{
  \providecommand{#1}{}
  \renewcommand{#1}{\textup{\textsc{#2}}\xspace}
}
\newclass{\pspace}{PSpace}
\newclass{\expspace}{ExpSpace}

\newcommand{\newcomponent}[2]{
  \providecommand{#1}{}
  \renewcommand{#1}{\ensuremath{\mathsf{#2}}\xspace}
}
\newcommand{\newinstruction}[2]{
  \providecommand{#1}{}
  \renewcommand{#1}{\ensuremath{\mathtt{#2}}\xspace}
}
\newcommand{\newrelation}[1]{%
  \expandafter\providecommand\csname #1\endcsname{}
  \expandafter\renewcommand\csname #1\endcsname{%
    \begingroup
    \ensuremath{\textcolor{#1}{\mathsf{#1}}}%
    \endgroup\xspace
  }%
}

\newcomponent{\rdmaprog}{R}
\newcomponent{\process}{P}
\newcomponent{\nodeset}{Node}
\newcomponent{\Tid}{Tid}
\newcomponent{\loc}{loc}
\newcomponent{\locset}{Loc}
\newcomponent{\valset}{Val} 
\newcomponent{\labset}{Lab}
\newcomponent{\pathset}{Seq}
\newcomponent{\A}{A}
\newcomponent{\B}{B}
\newcomponent{\C}{C}
\newcomponent{\D}{D}

\definecolor{rel}{RGB}{0,0,0}
\definecolor{po}{RGB}{50,75,100}
\definecolor{pf}{RGB}{153,102,51}
\definecolor{rf}{RGB}{0,153,0}
\definecolor{mo}{RGB}{255,128,0}
\definecolor{nfo}{RGB}{75,0,130}
\definecolor{rb}{RGB}{191,0,64}
\definecolor{sc}{RGB}{209,28,28}
\definecolor{ippo}{RGB}{0,173,238}
\definecolor{oppo}{RGB}{236,2,141}
\definecolor{ib}{RGB}{0,0,255}
\definecolor{ob}{RGB}{0,128,128}
\definecolor{sqp}{RGB}{0,0,0}
\newrelation{rel}
\newrelation{po}
\newrelation{pf}
\newrelation{rf}
\newrelation{mo}
\newrelation{nfo}
\newrelation{rb}
\newrelation{sc}
\newrelation{ippo}
\newrelation{oppo}
\newrelation{ib}
\newrelation{ob}
\newrelation{sqp}
\newcommand{\relName}[2]{\textup{\textcolor{#2}{\textsf{#1}}}}

\newinstruction{\nop}{skip}
\newinstruction{\assume}{assume}
\newinstruction{\CAS}{CAS}
\newinstruction{\poll}{poll}
\newinstruction{\rfence}{rfence}

\newinstruction{\rd}{rd}
\newinstruction{\wr}{wr}

\newinstruction{\nlR}{nlR}
\newinstruction{\nrW}{nrW}
\newinstruction{\nrR}{nrR}
\newinstruction{\nlW}{nlW}
\newinstruction{\nF}{nF}
\newinstruction{\lR}{lR}
\newinstruction{\lW}{lW}
\newinstruction{\CPU}{CPU}

\newinstruction{\Inst}{Inst}
\newinstruction{\nW}{nW}
\newinstruction{\writes}{W}
\newinstruction{\reads}{R}

\title{On the Verification Problem \\ of Remote Direct Memory Access programs}

\titlerunning{On the Verification Problem of RDMA programs}

\author{
  Parosh Aziz Abdulla\inst{1}\inst{2}\orcidlink{0000-0001-6832-6611} \and
  Mohamed Faouzi Atig\inst{1}\orcidlink{0000-0001-8229-3481} \and \\
  Govind Rajanbabu\inst{3}\orcidlink{0000-0002-1634-5893} \and
  Stephan Spengler \Envelope{}\inst{1}\orcidlink{0009-0009-5722-8843}
}

\institute{
  Uppsala University, Sweden\\
  \email{\{parosh.abdulla, mohamed\_faouzi.atig, stephan.spengler\}@it.uu.se}
  \and
  M\"alardalen University, Sweden
  \and
  Institute of Mathematical Sciences (IMSc), Chennai, India\\
  \email{govind@imsc.res.in}
}

\authorrunning{P.A. Abdulla, M.F. Atig, G. Rajanbabu, S. Spengler}

\begin{document}

\maketitle

\begin{abstract}
  Remote Direct Memory Access (RDMA) is a technology that allows direct memory access from the memory of one computer into that of another without involving either one's operating system. This enables high-throughput, low-latency networking, which is especially useful in massively parallel computer clusters.

  In this paper, we study the reachability and robustness problems for RDMA programs. 
  We show that reachability is undecidable in general, even for a restricted fragment of the model. 
  We then focus on robustness, which asks whether a program exhibits the same behaviours under the RDMA and sequential consistency (SC) semantics, and prove that this problem is decidable. 
  Our central technical result establishes a normal form for robustness violations, showing that any non-robust program admits a violating execution of a specific form.
  We then leverage this normal form to obtain a decision procedure that reduces robustness to reachability in finite-state programs with counters, yielding an \expspace upper bound in the general case, and a \pspace upper bound in the absence of poll operations. 
  Finally, we also show that both of these bounds are optimal. 

  \keywords{RDMA \and Weak Memory Models \and Robustness \and Verification \and Concurrent Systems.}
\end{abstract}

\section{Introduction}
\label{introduction:section}
Modern multiprocessor architectures implement weak memory models that relax the guarantees of sequential consistency (SC) to enable aggressive hardware and compiler optimisations. Classical models such as TSO, PSO, ARM, POWER, and RC11 permit various forms of reordering, buffering, and speculation, thereby admitting behaviours that are impossible under SC. While essential for performance, these relaxations significantly complicate reasoning about concurrent programs.
This increased complexity is inherent and well-established. For instance, reachability for finite-state concurrent programs is \pspace-complete under SC, but becomes non-primitive recursive under TSO and undecidable under the Release-Acquire (RA) fragment of RC11.
Thus, verifying safety properties under weak memory models is fundamentally harder than under SC, posing significant challenges for program analysis.

Remote Direct Memory Access (RDMA) has recently gained prominence in high-performance computing \cite{AmbalDEK0R24,DBLP:conf/esop/AmbalLR25}.
RDMA enables a process to directly read from or write to the memory of a remote machine, bypassing the remote CPU and operating system on the critical path.
This paradigm is widely used in data centres, where it offers low latency and high throughput. However, RDMA introduces a non-trivial memory model that differs fundamentally from both shared-memory concurrency and message passing.
RDMA supports one-sided operations, provides weak ordering guarantees, and allows remote memory updates to be observed asynchronously with respect to local CPU accesses. Consequently, reasoning about correctness requires explicit consideration of visibility, ordering, and synchronisation mechanisms such as fences and completion events.
These characteristics distinguish RDMA from existing weak memory models, necessitating specialised verification techniques.

In this work, we first study the reachability problem for RDMA programs.
We show that reachability is undecidable in general, via a reduction from Post's Correspondence Problem.
Notably, this undecidability holds even for finite-state RDMA programs without \emph{poll} operations or \emph{fence} instructions, highlighting the inherent complexity of the model.

The difficulty of proving safety properties under weak memory has motivated the study of correctness notions that enable sound reasoning despite relaxed semantics.
A fundamental question is whether a program behaves \emph{robustly} under a given weak memory model \cite{BouajjaniDM13}.
Informally, a program is robust if every behaviour under the weak model is observationally equivalent to some behaviour under SC.
Robustness is appealing for two reasons. First, if a program is robust and correct under SC, then it remains correct under the weak model.
Second, verifying robustness can be simpler than verifying safety directly.
For example, checking reachability under TSO is non-primitive-recursive \cite{AtigBBM10,AtigBBM12}, whereas checking robustness is \pspace-complete \cite{BouajjaniDM13}.

Recently, Ambal et al.~\cite{DBLP:conf/esop/AmbalLR25} initiated the study of robustness for RDMA programs.
They established sufficient conditions for robustness, but these are restrictive and do not characterise all robust programs (see, e.g., Fig.~6(b) in~\cite{DBLP:conf/esop/AmbalLR25}).
Their work leaves open the decidability and complexity of robustness.
In this work, we resolve these questions by showing that robustness of finite-state RDMA programs is decidable and \expspace-complete.
Our approach relies on a structural characterisation of robustness violations, enabling a reduction to reachability in SC programs augmented with finite control and counters.

To formalise RDMA behaviours, we build on the axiomatic semantics of Ambal et al.~\cite{AmbalDEK0R24}.
Their consistency check requires verifying acyclicity of two relations over execution events, which is rather complex.
We revisit this semantics and introduce a simpler, equivalent characterisation of RDMA consistency.
This refined view serves as the foundation for our decidability results.
To establish the \expspace upper bound, we introduce a novel instrumentation technique that transforms an RDMA program into an SC program.
This SC program simulates all behaviours of the original while preserving observational equivalence, enabling the use of existing SC verification tools.
We also identify a simpler fragment -- RDMA programs without poll operations -- for which robustness reduces to reachability in purely finite-state SC programs.

In summary, our main contributions are:
\begin{itemize}
    \item We prove that reachability for RDMA programs is undecidable, even without poll operations or fences.
    \item We simplify the axiomatic RDMA semantics of \cite{AmbalDEK0R24}, providing an equivalent but more tractable consistency check.
    \item We establish a normal form for robustness violations by rearranging computations around a pivotal event.
    \item We present a decision procedure for robustness based on an instrumentation that transforms RDMA programs into SC programs with counters, yielding an \expspace upper bound (\pspace without poll operations).
    \item We prove matching lower bounds: robustness is \expspace-complete in general and \pspace-complete without poll operations.
\end{itemize}

\noindent\textbf{Related work.}
We base our formalisation of RDMA behaviours on the RDMA semantics proposed by Ambal et al.~\cite{AmbalDEK0R24,DBLP:conf/esop/AmbalLR25}.
In~\cite{AmbalDEK0R24}, the authors assume TSO processors per node, while~\cite{DBLP:conf/esop/AmbalLR25} considers SC processors. 
Both works present axiomatic and operational semantics for RDMA and propose consistency-checking procedures. 
An earlier semantics by Dan et al.~\cite{DanLHV16} fails to capture certain RDMA behaviours, as shown in~\cite{AmbalDEK0R24}.

The reachability problem has been extensively studied for various weak memory models, such as 
TSO~\cite{AbdullaABN16,AbdullaAP15,AtigBBM12,AtigBBM10},
RA~\cite{AbdullaAAK19}, 
Causal Consistency~\cite{Toplas-LahavB22},  
PSO~\cite{AbdullaAKR15},
Power~\cite{AbdullaABDLM20}, 
Intel-x86 with Persistency~\cite{AbdullaABKS24,AbdullaABKS21},
Promising Semantics~\cite{AbdullaAGKV21}, 
and Localized Release-Acquire~\cite{SinghL24}.
The complexity of the reachability problem varies significantly across these models, ranging from non-primitive recursive to undecidable, depending on the specific memory model and program restrictions.
To the best of our knowledge, our work is the first to study the reachability problem for RDMA programs, establishing its undecidability in general.

The robustness problem has also been extensively studied for various memory models~\cite{DBLP:conf/icalp/BouajjaniMM11,BouajjaniDM13,CalinDMM13,DerevenetcM14,BouajjaniDM14,DBLP:phd/dnb/Derevenetc15}. 
For RDMA, Ambal et al.~\cite{DBLP:conf/esop/AmbalLR25} initiated the study of robustness, providing sufficient (but restrictive) conditions for ensuring robustness; their paper also gives robust examples not covered by these conditions (Fig.~6(b)).
To the best of our knowledge, our work is the first to establish decidability and complexity results for robustness and reachability-based safety verification under RDMA semantics.
\section{RDMA Program Semantics}\label{sec:prelims}

A program consists of a finite set of threads that operate on a distributed memory system composed of multiple nodes.
Each node contains its own local memory and processes, and nodes communicate with each other using Remote Direct Memory Access (RDMA) operations. 
We consider a finite instruction set consisting of local reads and writes, remote write operations, remote read operations, polling instructions ($\poll$), and a remote-fence instruction.
Local instructions operate directly on the node's memory, whereas remote instructions access the memory of other nodes via their \emph{network interface controllers} (NIC).

In the following, $x, y, z \in \locset$ denote memory locations, $v_r, v_w \in \valset$ denote values, and $n \in \nodeset = \set{1, \dots, N}$ denotes nodes.
A bar above a location or node indicates a remote location or node, respectively.
Each thread (or \emph{process}) $\process$ in an RDMA program is defined by the following grammar:
\begin{align*}
    \process ::= \quad &\ \nop \mid \process_1 ; \process_2 \mid \process_1 + \process_2 \mid \process^* \\
           \mid  &\ x := v_w \mid \assume(x = v_r) \mid \assume(x \neq v_r) \mid x := \CAS(z, v_r, v_w) \\
           \mid  &\ x := \bar y \mid \bar y := x \mid \rfence(\bar n) \mid \poll(\bar n)
\end{align*}

To model communication between nodes, each process $\process$ maintains a distinct buffer system (or \emph{queue pair} in the terminology of~\cite{AmbalDEK0R24}) for every remote node $\bar n$ whose memory it may access.
Remote operations issued by process $\process$ towards node $\bar n$ are first placed in the buffer system, which consists of multiple FIFO buffers.
Their arguments are not defined at the point when they are issued, since they may depend on local memory reads that occur after the remote operation is issued. 
Furthermore, they need to travel through the buffer system before they are propagated to the remote node.
Therefore, the completion of remote operations is not instantaneous, and may occur out of order with respect to the program order of the issuing process.
This allows the reordering of remote operations with respect to each other and with respect to local memory accesses, inducing a weak memory model.

We briefly describe the key buffers in the queue system; the full operational model involves additional buffers and is detailed in~\cite{AmbalDEK0R24}.
When a process issues a remote write, a symbolic entry $\bar{y} := x$ is placed in the first buffer $B_1$.
This entry is not yet instantiated with the value of $x$ -- the local variable $x$ may be modified by subsequent CPU operations before the entry is processed.
When the entry reaches the head of $B_1$, the current value $v$ of $x$ is read from local memory, and an instantiated entry $\bar{y} := v$ is forwarded to buffer $B_2$.
This delayed instantiation allows local writes to overtake the value read by pending remote operations.
Buffer $B_2$ transmits the instantiated write to the remote node, where it updates the target memory location $\bar{y}$.
Upon completion, an acknowledgment is placed in buffer $B_3$, which returns to the issuing node.
Remote reads follow a similar pattern, with instantiation occurring at the remote node.
Both remote read and remote write operations generate acknowledgments that can be consumed by $\poll$ operations.

$\poll$ operations provide a mechanism for synchronising with the completion of remote operations.
A poll instruction $\poll(\bar n)$ blocks until the earliest (in program order) pending remote operation issued by the calling process towards node $\bar n$ has completed.
$\poll$ operations can be used to prevent undesired reordering between remote operations and subsequent local memory accesses by ensuring that a remote operation has completed before execution proceeds.
However, note that each $\poll$ waits for exactly one remote operation to complete, rather than for all pending operations towards the same node.
A remote-fence instruction $\rfence$ ensures that all remote operations issued by the calling process towards a node prior to the fence are completed before any subsequent remote operations towards the same node are issued.

We will recall the full details of the axiomatic semantics in~\autoref{declarative:section}.
\section{Undecidability of Reachability}
\label{undecidability:section}
In this section, we show that the reachability problem for RDMA programs is undecidable.
Here, reachability is classical control-state reachability: given an RDMA program and a control state of one thread, we ask whether some execution from the initial configuration reaches a configuration in which that thread is in the target control state, while all other threads may be in arbitrary control states.
The proof will follow from a reduction from Post's Correspondence Problem (PCP)~\cite{Post1946}, a well-known undecidable problem.
An instance of PCP is given by two sequences of non-empty words over some finite alphabet $\Sigma$, say $u_1, u_2, \dots, u_N$ and $v_1, v_2, \dots, v_N$.
The problem asks whether there exists a sequence $i_1, i_2, \dots, i_k$ of indices such that $u_{i_1} u_{i_2} \dots u_{i_k} = v_{i_1} v_{i_2} \dots v_{i_k}$.
We will propose an RDMA program \rdmaprog consisting of four processes $\process_1$, $\process_2$, $\process_3$, and $\process_4$, each on a different node, such that there is an execution of \rdmaprog reaching a final configuration (i.e., all four processes terminate) if and only if PCP has a solution for the aforementioned instance.
\begin{wrapfigure}[13]{r}{0.4\textwidth}
    \vspace{-\baselineskip}
    \centering
    \includegraphics[width=0.4\textwidth,alt={Process 1 sends indices i_j to both process 2 and process 3. Process 2 sends the letters of u_i_j to process 4. Process 3 sends the letters of v_i_j to process 4.}]{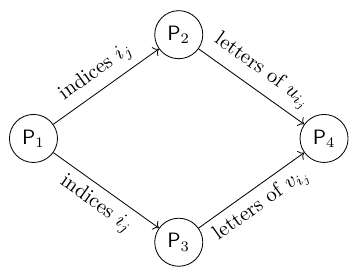}
    \caption{The PCP program.}
    \label{PCP:fig}
\end{wrapfigure}
The idea is that process $\process_1$ guesses a solution of PCP, which is a sequence of indices $i_1, i_2, \dots, i_k$, and sends the guessed indices to both $\process_2$ and $\process_3$.
For each received index $i$, process $\process_2$ ($\process_3$) sends the letters of $u_i$ ($v_i$) to process $\process_4$.
Process $\process_4$ checks that the sequences of letters received from $\process_2$ and $\process_3$ are identical.

In order to implement this idea, we need to ensure that the sequences of letters sent by the processes are received without any losses or duplications.
Since the set of letters is finite, the problem amounts to showing that we can transfer a sequence of letters over a finite alphabet from one process to another.
The main challenge is that remote writes propagate through the buffer system (cf.\ \autoref{sec:prelims}), and the receiving process has no built-in mechanism to detect when a new value arrives.
This can lead to missed or duplicated reads: the receiver may read before the value arrives, or read the same value multiple times.
We address this using an alternating bit protocol.

\autoref{noduplication:fig} illustrates the problem.
Process $\process_1$ wants to send a sequence of symbols from a finite alphabet $\{1, 2\}$ to process $\process_2$.
To send each symbol, process $\process_1$ uses a dedicated register for that symbol (register $a$ always holds $1$, register $b$ always holds $2$), and performs $\bar x := a$ or $\bar x := b$ accordingly.
Using constant values avoids complications from the delayed instantiation in buffer $B_1$ (see \autoref{sec:prelims}).
However, process $\process_2$ may not observe all sent symbols: in \autoref{noduplication:fig}, process $\process_1$ sends the sequence $1, 2$, but $\process_2$ may read only the final value $2$, missing the intermediate $1$.
To fix this synchronisation problem, we use an alternating bit protocol (ABP) between each pair of processes.
The ABP ensures that each sent symbol is received exactly once by the receiving process.
\begin{figure}[t!]
    \includegraphics[width=\textwidth,alt={Process 1 sends messages x = a and x = b to process 2. They are instantiated as x = 1 and x = 2. The first message is received by process 2, which changes the value of x to 1 and generates an acknowledgement. Then, the second message is received by process 2 which changes the value of x to 2 and generates another acknowledgement.}]{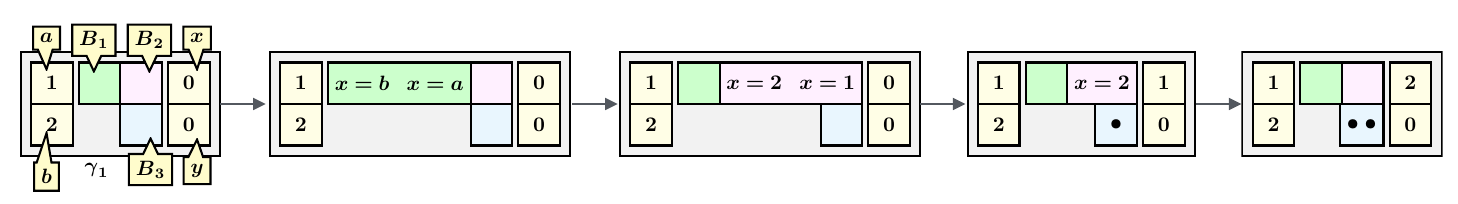}
    \caption{Message loss due to unsynchronised reads.}
    \label{noduplication:fig}
\end{figure}
\begin{figure}[b!]
    \includegraphics[width=\textwidth,alt={The configurations of the alternating bit protocol, where one process sends the message x = 2 followed by x = 1, whithout message loss or duplication.}]{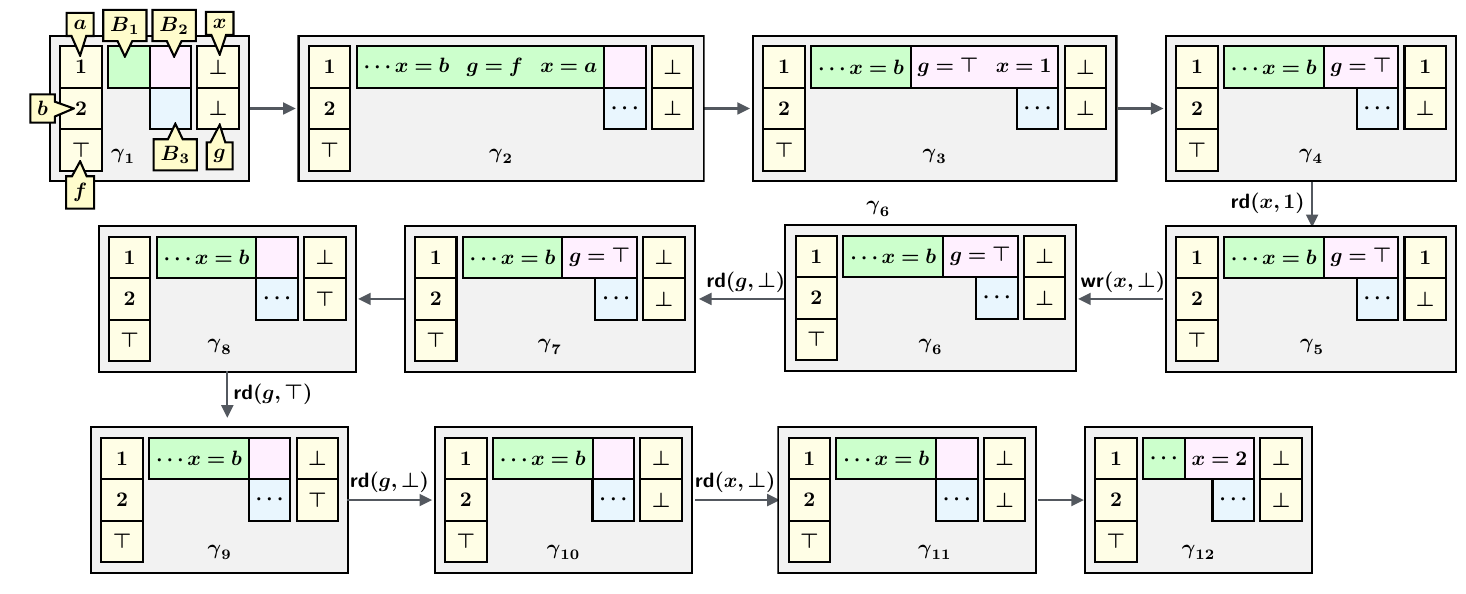}
    \caption{The Alternating Bit Protocol.}
    \label{AB:fig}
\end{figure}
The protocol uses two extra registers $f$ and $g$ at the sender and receiver processes, respectively.
Register $f$ at the sender process holds the alternating bit, while register $g$ at the receiver process observes the alternating bit.
The key invariant is that each symbol is accepted only when $g$ changes to the sender's bit, which prevents both losses and duplicate deliveries.
The protocol works as follows (see~\autoref{AB:fig}).
The initial value of $f$ and $g$ are $\top$ and $\bot$, respectively, while $a$ and $b$ hold the values from the alphabet (here $1$ and $2$).
In \autoref{AB:fig}, we assume that process $\process_1$ wants to send a word over the alphabet $\{1,2\}$ to process $\process_2$, of the form $\cdots21$, i.e., the word starts with $1$ followed by $2$.
The process $\process_1$ first performs $\bar x := a$ to send the value $1$ to $\process_2$, and then $\bar g := f$ to set the flag $g$ to $\top$, indicating that a new value $1$ has been sent (see configuration $\gamma_2$).
When these two messages reach the head of buffer $B_1$, they acquire their values ($1$ resp. $\top$), and are forwarded to buffer $B_2$ (cf. configuration $\gamma_3$).
When $x=1$ is at the head of buffer $B_2$, the node of $\process_2$ updates its local variable $x$ to $1$ (cf. configuration $\gamma_4$).
Once process $\process_2$ has read the value $1$, we need to ensure that it can detect and read the next value.
To that end, process $\process_2$ carries out a sequence of operations, where it first performs the local write operation $\wr(x, \bot)$ to reset $x$ to its initial value $\bot$.
Then it checks that the receiver-side flag is still unset by performing the local read operation $\rd(g, \bot)$ (cf. configuration $\gamma_5$).
This ensures that no additional write operations were performed in the previous steps.
Next, process $\process_2$ performs the local read operation $\rd(g, \top)$ to detect the next fresh delivery and ensure that the old value $1$ will not be read again (cf. configuration $\gamma_8$).
Finally, it resets the flag $g$ to its initial state $\bot$.

Using this protocol, we can implement the PCP reduction as follows (details can be found in Appendix~\ref{sec:code}).
Process $\process_1$ uses the ABP to send the guessed indices to both $\process_2$ and $\process_3$.
For each received index $i$, process $\process_2$ uses the ABP to send the letters of $u_i$ to $\process_4$, while $\process_3$ sends the letters of $v_i$.
Process $\process_4$ checks that the two received letter sequences are identical, alternating between receiving from $\process_2$ and $\process_3$.
The buffered communication is essential: since $\process_2$ and $\process_3$ may produce letters at different rates, the buffers decouple their execution, allowing the final sequences to match even when the corresponding words $u_i$ and $v_i$ have different lengths.
The program reaches a final configuration if and only if PCP has a solution.
This establishes the following result.

\begin{theorem}
\label{thm:undecidability}
    The reachability problem for RDMA programs is undecidable.
\end{theorem}

We note that the reduction uses only a restricted fragment of the RDMA model: it requires only local operations and remote writes, without remote reads or any synchronisation primitives such as $\poll$ or $\rfence$.
This demonstrates that reachability is undecidable even for this minimal fragment.

\section{The Declarative Semantics of RDMA}
\label{declarative:section}

We build upon the declarative semantics for RDMA developed by Ambal et al.~\cite{AmbalDEK0R24,DBLP:conf/esop/AmbalLR25}.
The CPU part of our model and the general notion of sequential consistency follow standard definitions from the literature.
The RDMA-specific aspects -- NIC operations, their labels and events, \pf, \nfo and derived relations, and the RDMA consistency conditions -- are adopted from their work.

An RDMA program $\rdmaprog$ is a collection of threads $(\process_t)_{t \in \Tid}$, where $\Tid = \set{1, \dots, M}$ is the set of thread identifiers.
We assume a function $n: \locset \to \nodeset$ that maps variables to their node.
An execution of a program generates events, which are labelled by their type.
Let $\labset$ consist of all labels of the following forms:
\begin{align*}
    & \text{NIC remote read:}   && \nrR(\bar y, v_r)    \quad  && \text{CPU local read:}       && \lR(x, v_r)          \\
    & \text{NIC local write:}   && \nlW(x, v_w, \bar n) \quad  && \text{CPU local write:}      && \lW(x, v_w)          \\
    & \text{NIC local read:}    && \nlR(x, v_r, \bar n) \quad  && \text{CPU compare-and-swap:} && \CAS(z, v_r, v_w)    \\
    & \text{NIC remote write:}  && \nrW(\bar y, v_w)    \quad  && \text{CPU poll:}             && \poll(\bar n)        \\
    & \text{NIC fence:}         && \nF(\bar n)
\end{align*}
An \emph{event} is defined as a triple $e = \tuple{\iota, t, l}$ where $\iota \in \Nat$ is the program counter, $t \in \Tid$ is the thread id, and $l \in \labset$ is a label.
We omit the program counter $\iota$ and the thread id $t$ when they are clear from the context, i.e. the event is identified by its label.
Some operations generate multiple events.
A remote read operation $x := \bar y$ generates a NIC remote read $\nrR(\bar y, v)$ followed by a NIC local write $\nlW(x, v, n(\bar y))$.
A remote write operation $\bar y := x$ works similarly: it generates a NIC local read $\nlR(x, v, n(\bar y))$ followed by a NIC remote write $\nrW(\bar y, v)$.
A compare-and-swap operation $x := \CAS(z, v_r, v_w)$ generates either (on success) $\CAS(z, v_r, v_w) \cdot \lW(x, v_r)$ or (on failure) $\lR(z, v) \cdot \lW(x, v)$ where $v \neq v_r$.
We inductively define $\pathset(\process) \subseteq \labset^*$ as the set of all finite label sequences that can be generated by the process $\process$.
The formal definition is standard and therefore omitted, but can be found in Appendix~\ref{sec:def_seq}.

\paragraph{Notation.}
We use $\iota(e)$, $t(e)$, $n(e)$, $\bar n(e)$, $\loc(e)$, $v_r(e)$ and $v_w(e)$ to refer to the program counter, thread id, local node, remote node, location, read value and write value of $e$, whenever applicable.

Let $E$ be a set of events.
With $E.\nrR, E.\nlW, \dots$ we refer to the subset of events of $E$ that have label $\nrR, \nlW, \dots$, respectively.
We define:
\begin{alignat}{3}
    &E.\CPU     &&= E.\lR \cup E.\lW \cup E.\CAS \cup E.\poll   \tag*{(CPU events)}         \\
    &E.\nW      &&= E.\nlW \cup E.\nrW                          \tag*{(NIC write events)}   \\
    &E.\writes  &&= E.\nlW \cup E.\nrW \cup E.\lW \cup E.\CAS   \tag*{(write events)}       \\
    &E.\reads   &&= E.\nlR \cup E.\nrR \cup E.\lR \cup E.\CAS   \tag*{(read events)}        \\
    &E.\Inst    &&= E \setminus (E.\nlW \cup E.\nrW)            \tag*{(instantaneous events)}
\end{alignat}
Given a relation $\rel \subset E \times E$ and a subset $E' \subset E$, we use $\rel^{-1} = \set{(e_2, e_1) \mid (e_1, e_2) \in \rel}$, $[E'] = \set{(e, e) \mid e \in E'}$ and $\rel\mid_{E'} = \rel \cap E' \times E'$.
We usually write $e_1 \to[\rel] e_2$ if $(e_1, e_2) \in \rel$.
If $\rel' \subset E \times E$ is another relation, then $\rel;\rel' = \set{(e_1, e_3) \mid \exists\ e_2 \in E: (e_1, e_2) \in \rel, (e_2, e_3) \in \rel'}$.
Lastly, we define the \emph{same-queue-pair} relation $\sqp = \set{ (e_1, e_2) \in E \times E \mid t(e_1) = t(e_2) \land \bar n(e_1) = \bar n(e_2)}$.

\begin{definition}[Trace]
    \label{def:trace}
    A \emph{trace} of an RDMA program $\rdmaprog = (\process_t)_{t \in \Tid}$ is a tuple $T = \tuple{ E, \po, \pf, \rf, \mo, \nfo }$ where:
    \begin{itemize}
        \item
        $E = E^0 \cup \bigcup_{t \in \Tid} E^t$ is a finite set of events such that:
        \begin{itemize}
            \item $E^0 = \set{ (\_, \_, \lW(x, v_x)) \mid x \in \locset }$ is the set of \emph{initialisation events}, where $v_x$ is the \emph{initial value} of $x$.
            \item
            For each $t \in \Tid$, there is $s_t \in \pathset(\process_t)$ and $0 \leq n \leq |s_t|$
            such that $E^t = \set{\tuple{\iota, t, s_t[\iota]} \mid 1 \leq \iota \leq n}$
            (i.e., the events of each thread form a prefix of some path in $\pathset(\process_t)$).
        \end{itemize}
        
        \item
        $\po \subset \bigcup_{t\in\Tid}(E^t \times E^t)$ is the \relName{program order}{po} relation on $E$, which is a strict total order on the events of each thread,
        uniquely defined such that $(e_1, e_2) \in \po$ if and only if $\iota(e_1) < \iota(e_2)$ and $t(e_1) = t(e_2)$.

        \item
        $\pf \subset (E.\nW \times E.\poll) \cap \po \cap \sqp$ is the \relName{polls-from}{pf} relation on $E$, that relates NIC write events to later (w.r.t. \po) poll events on the same queue pair.
        It is uniquely defined such that every poll event polls from exactly one NIC write, which must be the earliest (w.r.t. \po) non-polled one.
        That is, $(e_1, e_2) \in \pf$ if and only if $e_1 \in E.\nW$, $e_2 \in E.\poll$, $(e_1, e_2) \in \po\cap\sqp$ and:
        $$ \Sizeof{\Set{ e_1' \in E.\nW \mid (e_1', e_1) \in \po\cap\sqp }} = \Sizeof{\Set{ e_2' \in E.\poll \mid (e_2', e_2) \in \po\cap\sqp }} $$

        \item
        $\rf \subset E.\writes \times E.\reads$ is the \relName{reads-from}{rf} relation on $E$, which relates write events to read events on the same variable and with the same value.
        That is, $(e_1, e_2) \in \rf$ only if $\loc(e_1) = \loc(e_2)$ and $v_w(e_1) = v_r(e_2)$.
        Furthermore, each read event reads from exactly one write event.

        \item
        $\mo = \bigcup_{x \in \locset} \mo_x \subset E.\writes \times E.\writes$ is the \relName{modification order}{mo} relation on $E$.
        Each $\mo_x$ is a strict total order on the write events on location $x$, such that $E^0_x \times (E.\writes_x \setminus E^0_x) \subset \mo_x$,
        where $E^0_x = \set{(\_, \_, \lW(x, v_x))} \subseteq E^0$ and $E.\writes_x = \set{e \in E.\writes \mid \loc(e) = x}$.
        
        \item
        $\nfo \subset \sqp$ is the \relName{NIC flush order}{nfo} relation on $E$, which relates NIC events on the same queue pair in the following way:
        For all $(e_1, e_2) \in \sqp$, if $e_1 \in E.\nlR$ and $e_2 \in E.\nlW$, then exactly one of the pairs $(e_1, e_2)$ or $(e_2, e_1)$ is in $\nfo$.
        The same holds true for $e_1 \in E.\nrR$ and $e_2 \in E.\nrW$.
    \end{itemize}
\end{definition}

\begin{definition}[Derived Relations]
    Given a trace $T = \tuple{ E, \po, \pf, \rf, \mo, \nfo }$, we define:
    \begin{align*}
        \rb     =&\ (\rf^{-1};\mo)                                                                           \tag*{(\relName{reads-before}{rb})} \\
        \sc     =&\ (\po \cup \rf \cup \rb \cup \mo)^+                                                       \tag*{(\relName{sequential consistency}{sc})} \\
        \ippo    &\subset \po \text{ as in \autoref{fig:ippo-oppo}}                                          \tag*{(\relName{issue-preserved program order}{ippo})} \\
        \oppo    &\subset \po \text{ as in \autoref{fig:ippo-oppo}}                                          \tag*{(\relName{observation-preserved program order}{oppo})} \\
        \ib     =&\ (\ippo \cup \rf \cup \pf \cup \nfo)^+                                                    \tag*{(\relName{issued-before}{ib})} \\
        \ob     =&\ (\oppo \cup \rf \cup ([\nlW]; \pf) \cup \nfo \cup \rb \cup \mo \cup ([\Inst]; \ib))^+    \tag*{(\relName{observed-before}{ob})}
    \end{align*}
\end{definition}

The \rb relation ensures that if $w_1 \to[\rf] r$ and $w_1 \to[\mo] w_2$, then $r$ was observed before $w_2$.
To build intuition for the remaining relations, recall that NIC writes are non-instantaneous (they may be buffered), while events in $\Inst$ are instantaneous.
Hence we distinguish \emph{issue} order from \emph{observation} order.

\begin{figure}[bp]
    \centering
    \renewcommand{\tabcolsep}{6pt}
    \begin{tabular}{c|cccccc}
        \po \tikz[scale=0.2] \draw[->] (0,0) .. controls (0,1) .. (1.5,1);
                        & \CPU          & \nlR          & \nrW          & \nrR          & \nlW          & \nF           \\
        \hline
        \CPU            & \checkmark    & \checkmark    & \checkmark    & \checkmark    & \checkmark    & \checkmark    \\
        \nlR            & $\times$      & \sqp          & \sqp          & \sqp          & \sqp          & \sqp          \\
        \nrW            & $\times$      & $\times$      & \sqp          & \sqp          & \sqp          & \ippo: \sqp   \\
        \nrR            & $\times$      & $\times$      & $\times$      & $\times$      & \sqp          & \sqp          \\
        \nlW            & $\times$      & $\times$      & $\times$      & $\times$      & \sqp          & \ippo: \sqp   \\
        \nF             & $\times$      & \sqp          & \sqp          & \sqp          & \sqp          & \sqp
    \end{tabular}
    \caption{The relations \ippo and \oppo as subsets of \po.}
    \label{fig:ippo-oppo}
\end{figure}

The relation \ippo captures \po-edges that must be preserved at issue points, whereas \oppo captures those that must be preserved at observation points.
Thus \oppo is closer to the usual notion of preserved program order.
In \autoref{fig:ippo-oppo}, a \checkmark\ means that \po-edges between the two event types are contained in both \ippo and \oppo, whereas $\times$ means that \po-edges are not contained in either relation.
\sqp means that \po-edges between two events on the same queue pair are preserved in both relations and \ippo: \sqp means that these edges are only preserved in \ippo but not \oppo.
For example, a NIC write followed by a fence on the same queue pair is in \ippo but not necessarily in \oppo: the write is issued first, but may be observed later due to buffering.

The relation \ib collects ordering constraints at issue points (including \ippo, \rf, \pf, and \nfo), while \ob is the corresponding observation-level causality relation (including \oppo, \rf, $([\nlW];\pf)$, \nfo, \rb, \mo, and $([\Inst];\ib)$).
Intuitively, $([\Inst];\ib)$ lifts issue-order constraints to observation order when the source event is instantaneous.

We use the notation '$T.$' to refer to the various components of $T$, e.g. $T.E$ is the event set of $T$ and $T.\ib$ is the \ib-relation of $T$.

\begin{definition}[Consistency and Robustness]
    A trace $T$ is called
    \begin{itemize}
        \item
        \parbox{4cm}{\emph{SC-consistent}}
        if $T.\sc$ is acyclic (irreflexive).
        
        \item
        \parbox{4cm}{\emph{RDMA-consistent \cite{DBLP:conf/esop/AmbalLR25}}} 
        if both $T.\ib$ and $T.\ob$ are acyclic (irreflexive).
    \end{itemize}
    A program \rdmaprog is said to be \emph{robust} if all its RDMA-consistent traces are also SC-consistent.
    We call a trace \emph{violating} if it is RDMA-consistent but not SC-consistent.
\end{definition}

The definition of SC consistency is standard.
The definition of RDMA consistency is well motivated by Ambal et al. who show that the operational and declarative semantics of RDMA are equivalent.
Our first result is that the first condition of RDMA-consistency, acyclicity of the issued-before relation, is redundant:

\begin{theorem}
    If \ib is cyclic, then \ob is cyclic.
\end{theorem}
\begin{proof}
    Consider a cycle in \ib, i.e. a sequence of events $e_0, e_1, \dots, e_n = e_0$ such that $e_{i-1} \to[\rel_i] e_i$ and $\rel_i \in \set{ \ippo, \rf, \pf, \nfo }$ for all $1 \leq i \leq n$.
    If all edges of the cycle are in \ob, we are done.
    So, assume that there exists at least one edge $e_{i-1} \to[\rel_i] e_i$ that is not contained in \ob.
    In particular, it holds that:
    \begin{multline*}
        (e_{i-1}, e_i)
        \in (\ippo \cup \rf \cup \pf \cup \nfo) \setminus \ob 
        \subseteq (\ippo \setminus \oppo) \cup ([\nrW];\pf) \\
        = ([\nrW];\po;[\nF]) \cup ([\nlW];\po;[\nF]) \cup ([\nrW];\pf)
    \end{multline*}
    The event $e_i$ must be a NIC fence or a poll event, which are both instantaneous.
    We conclude $(e_i, e_i) \in \ib$ and since $([\Inst]; \ib) \subseteq \ob$, it follows that $(e_i, e_i) \in \ob$, which means that \ib is cyclic.
\end{proof}

Henceforth, we only need to check acyclicity of $\ob$ when verifying RDMA-consistency.

\section{Normal Form for Violating Computations}
\label{sec:normal_form}

Our main result is that every non-robust RDMA program has a violating trace in a special \emph{normal form}:

\begin{theorem}[Normal Form]
\label{thm:normal_form_main}
    If an RDMA program \rdmaprog has a violating trace, then it has a violating trace whose event set can be partitioned as $E = E_1 \cup E_2$ such that:
    \begin{enumerate}
        \item The events within $E_1$ and $E_2$ are in an SC-consistent order, respectively. That is, $\sc|_{E_1}$ and $\sc|_{E_2}$ are acyclic.
        \item All \ob-edges between $E_1$ and $E_2$ go from $E_1$ to $E_2$, that is, for all $e_1 \in E_1$ and $e_2 \in E_2$, if $e_1$ and $e_2$ are related by \ob, then $e_1 \to[\ob] e_2$.
    \end{enumerate}
\end{theorem}

The key insight is that we can simplify any violating trace by identifying a \emph{maximal event} -- one that could be the last issued event -- and reorganising the execution around it.
This reorganisation removes unnecessary delays while preserving both RDMA-consistency and the SC-violation.

This chapter presents the construction of the normal form trace and establishes the correctness of \autoref{thm:normal_form_main}.
The next chapter shows how this normal form enables a practical robustness decision procedure: we can systematically search for violating traces by building up $E_1$ and $E_2$ in parallel during a single SC-simulation, non-deterministically assigning each generated event to either $E_1$ or $E_2$.

\subsection{Example}
\label{sec:example}
\begin{wrapfigure}[6]{r}{0.3\textwidth}
    \vspace{-2\baselineskip}
    \setlength{\tabcolsep}{4pt}
    \centering
    \begin{tabular}{l l||l}
                    & $x = 1$ & $z = 0$ \\
                    & $y = 0$ & $w = 1$ \\
        \hline
        \texttt{1:} & $y := \bar w$ & \\
        \texttt{2:} & $\bar z := x$ & \\
        \texttt{3:} & $x := 2$      & \\
    \end{tabular}
\end{wrapfigure}
Consider the program \rdmaprog on the right.
It has two nodes with two variables each: node 1 has variables $x$ and $y$, node 2 has variables $z$ and $w$.
We consider only a single thread on node 1, with three instructions: a remote read of $w$ into $y$, a remote write of $x$ into $z$, and a local write of $2$ into $x$.
This program is not robust: The read on $x$ in line 2 can be delayed after the write on $x$ in line 3 due to the RDMA memory model, causing it to read the value $2$ instead of the initial value $1$.
This gives rise to a trace $T = \tuple{ E, \po, \pf, \rf, \mo, \nfo }$ with event set
$$E = E_0~\cup~\set{ \nrR(\bar w, 1), \nlW(y, 1, \bar 2), \nlR(x, 2, \bar 2), \nrW(\bar z, 2), \lW(x, 2) }$$
and relations:

\noindent\makebox[\textwidth][c]{%
    \begin{tikzpicture}[xscale=2.7,yscale=-2]
        \node (e1) at (1,0) {$\nrR(\bar w, 1)$};
        \node (e2) at (2,0) {$\nlW(y, 1, \bar 2)$};
        \node (e3) at (3,0) {$\nlR(x, 2, \bar 2)$};
        \node (e4) at (4,0) {$\nrW(\bar z, 2)$};
        \node (e5) at (5,0) {$\lW(x, 2)$};
        
        \draw[->] (e1) -- node[above] {\ib, \ob} node[below] {\po, \sc} (e2);
        \draw[->] (e2) to[bend left] node[below] {\po, \sc} (e3);
        \draw[->] (e3) to[bend left] node[below] {\nfo} node[above] {\ib, \ob} (e2);
        \draw[->] (e3) to[bend left] node[below] {\po, \sc} node[above] {\ib, \ob} (e4);
        \draw[->] (e4) to[bend left] node[below] {\po, \sc} (e5);
        \draw[->] (e5) to[bend left=45] node[below] {\rf, \sc} node[above] {\ib, \ob} (e3);
    \end{tikzpicture}%
}
The set of initial events $E_0$ is not shown in the figure, but since these events have only outgoing edges, they cannot be part of any cycle.
$T$ is RDMA-consistent, since $\ob$ is acyclic, and violating, since $\sc$ is cyclic.
Furthermore, it is minimal, since removing the last event (w.r.t. \po) makes it non-violating.
Note that we cannot remove any other event, since this would not yield a valid trace of \rdmaprog.

We choose the maximal event $e = \lW(x, 2)$ and consider one possible ordering of when events happen, called a \emph{linearisation} $\tau = \tau_1~e~\tau_2$:
$$\tau = \nrR(\bar w, 1)~\lW(x, 2)~\nlR(x, 2, \bar 2)~\nrW(\bar z, 2)~\nlW(y, 1, \bar 2)$$
Removing $e$ from $\tau$ and adjusting the values of events that depended on it yields the event sequence $\tau'$:
$$\tau' = \nrR(\bar w, 1)~\nlR(x, 1, \bar 2)~\nrW(\bar z, 1)~\nlW(y, 1, \bar 2)$$
The corresponding trace $T'$ on the event set $E' = E\setminus\set{e}$ is:

\noindent\makebox[\textwidth][c]{%
    \begin{tikzpicture}[xscale=3,yscale=-2]
        \node (e1) at (1,0) {$\nrR(\bar w, 1)$};
        \node (e2) at (2,0) {$\nlW(y, 1, \bar 2)$};
        \node (e3) at (3,0) {$\nlR(x, 1, \bar 2)$};
        \node (e4) at (4,0) {$\nrW(\bar z, 1)$};
        
        \draw[->] (e1) -- node[above] {\ib, \ob} node[below] {\po, \sc} (e2);
        \draw[->] (e2) to[bend left] node[below] {\po, \sc} (e3);
        \draw[->] (e3) to[bend left] node[below] {\nfo} node[above] {\ib, \ob} (e2);
        \draw[->] (e3) -- node[below] {\po, \sc} node[above] {\ib, \ob} (e4);
    \end{tikzpicture}%
}
We observe that the values of the read and write events of the instruction $z := x$ have been adjusted, since they now read and write the initial value instead of the value written by the removed event.
$T'$ is both RDMA-consistent and SC-consistent, since $T'.\ob$ and $T'.\sc$ are acyclic, respectively.

Since $T'$ is SC-consistent, we can find an ordering of its events that respects all SC constraints:
$$\sigma = \nrR(\bar w, 1)~\nlW(y, 1, \bar 2)~\nlR(x, 1, \bar 2)~\nrW(\bar z, 1)$$
We now construct the normal form trace $T''$ by inserting the maximal event $e$ back into this SC-consistent ordering.
Specifically, we place events that happened before $e$ in $\tau$ (only $\nrR(\bar w, 1)$ in this case), then $e$ itself, then the remaining events in their SC-consistent order:
$$\tau'' = \nrR(\bar w, 1)~\lW(x, 2)~\nlW(y, 1, \bar 2)~\nlR(x, 2, \bar 2)~\nrW(\bar z, 2)$$
The linearisation $\tau''$ corresponds to the trace $T''$:

\noindent\makebox[\textwidth][c]{%
    \begin{tikzpicture}[xscale=2.7,yscale=-2]
        \node (e1) at (1,0) {$\nrR(\bar w, 1)$};
        \node (e2) at (2,0) {$\nlW(y, 1, \bar 2)$};
        \node (e3) at (3,0) {$\nlR(x, 2, \bar 2)$};
        \node (e4) at (4,0) {$\nrW(\bar z, 2)$};
        \node (e5) at (5,0) {$\lW(x, 2)$};
        
        \draw[->] (e1) -- node[above] {\ib, \ob} node[below] {\po, \sc} (e2);
        \draw[->] (e2) -- node[below] {\po, \sc, \nfo} node[above] {\ib, \ob} (e3);
        \draw[->] (e3) to[bend left] node[below] {\po, \sc} node[above] {\ib, \ob} (e4);
        \draw[->] (e4) to[bend left] node[below] {\po, \sc} (e5);
        \draw[->] (e5) to[bend left=45] node[below] {\rf, \sc} node[above] {\ib, \ob} (e3);
    \end{tikzpicture}%
}
It is RDMA-consistent and uses the same set of events as $T$.
Furthermore, it has the same SC relation as $T$ and is therefore also violating.

Note that although $T$ and $T''$ have exactly the same events, their ordering relations differ.
In particular, the $\nfo$-relation between the NIC local write $\nlW(y, 1, \bar 2)$ and the NIC local read $\nlR(x, 1, \bar 2)$ has been reversed.
Intuitively, the normal form eliminates the unnecessary delay of $\nlW(y, 1, \bar 2)$ after the read and write events of the later instruction $z := x$.

\subsection{Definitions and Properties}
\begin{definition}[Linearisation of a Trace]
\label{def:linearisation_of_trace}
    A \emph{linearisation} of an RDMA-consistent trace $T = \tuple{ E, \po, \pf, \rf, \mo, \nfo }$ is a strict total order on the events $E = \set{e_1, \dots, e_n}$ that is an extension (a superset) of $T.\ob$.
    We usually write it as a sequence $\tau = e_1~\dots~e_n$.
    Similarly, a \emph{linearisation} of an SC-consistent trace $T$ is a strict total order extending $T.\sc$.
\end{definition}
Given a sequence of events $\tau = e_1~\dots~e_n$, we write $e_i <_\tau e_j$ if $i < j$.
Furthermore, we define $T\of\tau = \tuple{ E, \po, \pf, \rf, \mo, \nfo }$, where:
\begin{itemize}
    \item $E = \set{ e_1, \dots, e_n}$
    \item $\po$ and $\pf$ are defined as in \autoref{def:trace} with respect to the event set $E$
    \item $e \to[\rf] e'$ if and only if $e <_\tau e'$, $e$ is a write event, $e'$ is a read event on the same variable, and there is no other write event on the same variable between $e$ and $e'$ in $\tau$
    \item $e \to[\mo] e'$ if and only if $e <_\tau e'$ and both $e$ and $e'$ are write events on the same variable
    \item $e \to[\nfo] e'$ if and only if $e <_\tau e'$, both $e$ and $e'$ are events belonging to the same queue pair, exactly one of them is a read event and one of them is a write event, and they are either both local or both remote NIC events
\end{itemize}
$T(\tau)$ is a (not necessarily RDMA-consistent) trace of \rdmaprog if and only if the event set $E$, the polls-from relation \pf and the reads-from relation \rf are \emph{well-formed}, meaning they are in the form as required in \autoref{def:trace}.
In particular, there must be at least as many NIC writes as polls on each queue pair at any time, and each pair of events related by \rf must have matching values.
All other requirements for relations are automatically fulfilled.

\begin{definition}[Maximal Event]
\label{def:maximal_event}
    Given a trace $T$, an event $e \in E$ is \emph{maximal} if it is $\po$-maximal and there is no CPU-event $e'$ such that $e \to[\ib] e'$.
\end{definition}
\begin{lemma}
    Any RDMA-consistent trace $T$ has at least one maximal event $e$.
\end{lemma}
\begin{proof}
    Assume that there is no maximal event.
    Construct a sequence of events $e_1, e'_1, e_2, e'_2, \dots$ as follows.
    Let $e_1$ be any $\po$-maximal event.
    For all $i \in \Nat$, let $e'_i$ be a CPU-event such that $e_i \to[\ib] e'_i$ and let $e_{i+1}$ be a $\po$-maximal event such that $e'_i \to[\po^*] e_{i+1}$.
    Because there are only finitely many $\po$-maximal events, this sequence must eventually repeat and form a cycle.
    Since $\po$-edges from CPU-events to later events are in $\ippo$ and therefore in $\ib$, the cycle is completely contained in $\ib$, which contradicts that $T$ is RDMA-consistent.
\end{proof}

\subsection{Construction}
Consider an RDMA program \rdmaprog that is not robust.
Choose a violating trace $T$ that is minimal in the sense that any trace with fewer events is not violating.
In particular, $T = \tuple{ E, \po, \pf, \rf, \mo, \nfo }$ is RDMA-consistent and $T.\sc = \tuple{ E, \po, \rf, \rb, \mo }$ is cyclic.
Let $e$ be a maximal event and let $\tau = \tau_1~e~\tau_2$ be a linearisation of $T$.
We remove the maximal event $e$ from $\tau$ to obtain $\tau'$ in the following way.
For each event $\bar e \in E \setminus \set{e}$ in $\tau$:
\begin{itemize}
    \item If $\bar e$ is a NIC local read $\nlR(x, v_r, \bar n)$ or a NIC remote read $\nrR(\bar y, v_r)$,
    we replace the value $v_r$ with the value of the last write event on $x$ or $\bar y$, respectively.
    \item If $\bar e$ is a NIC remote write $\nrW(\bar y, v_w)$ or a NIC local write $\nlW(x, v_w, \bar n)$,
    we replace the value $v_w$ with the value of the local or remote read coming immediately before $\bar e$ in $\po$.
    \item Otherwise we keep the event $\bar e$ unchanged.
\end{itemize}

We claim that $T' = T(\tau')$ is an RDMA-consistent trace of \rdmaprog.
Since $T$ was minimal, it follows that $T'$ is not violating and therefore $T'.\sc$ is acyclic.
Let $\sigma$ be a linearisation of $T'.\sc$.
Define $\tau'' = \tau_1\downarrow_\sigma~e~\tau_2\downarrow_\sigma$
where $\tau_i\downarrow_\sigma$ are the events of $\tau_i$ in the order in which they appear in $\sigma$.

\begin{theorem}
\label{thm:normal_form}
    $T'' = T(\tau'')$ is a violating trace of \rdmaprog.
\end{theorem}
The full proof can be found in Appendix~\ref{sec:proofs}.
It proceeds in two phases, corresponding to the two trace transformations performed on the original trace $T$.
In the first phase, we establish that after removing the maximal event $e$ from $\tau$ and adjusting dependent NIC read/write values, the resulting trace $T'$ remains both well-formed and RDMA-consistent.
In the second phase, we analyse the trace $T''$ obtained by reinserting $e$ into $T'$.
The key structural observation is that several fundamental relations (\po, \ippo, \oppo, \pf, \rf, \mo, \rb) remain unchanged between $T$ and $T''$.
Building on this, we prove that $T''$ is RDMA-consistent by showing all $\ob''$-edges follow the order given by $\tau''$, that is, they relate events earlier in $\tau''$ to events later in $\tau''$.
Finally, we establish that $T.\sc = T''.\sc$, completing the proof that $T''$ violates the same SC properties as $T$ while being RDMA-consistent.

This establishes \autoref{thm:normal_form_main}: the partition required by the theorem is given by $E_1 = \set{e} \cup \set{e' \in E'' \mid e' \in \tau_1}$ and $E_2 = \set{e' \in E'' \mid e' \in \tau_2}$.
The SC-consistency within each partition follows from the construction, which orders events in $\tau_1$ and $\tau_2$ according to the SC-linearisation $\sigma$ of $T'$.
Note that due to maximality, $e$ is \po-later than any event in $\tau_1$, which means it can be added to $E_1$ while preserving SC-consistency.
The directionality of \ob-edges between the two partitions is guaranteed by the proof that all $\ob''$-edges respect the ordering of $\tau''$.
\section{An EXPSPACE-Algorithm for Robustness}

\subsection{Overview}

The previous section established that every non-robust RDMA program has a violating trace whose event set can be partitioned as $E = E_1 \cup E_2$ where events within each partition are SC-consistent and all \ob-edges between partitions go from $E_1$ to $E_2$.
This chapter leverages this normal form to establish an upper bound for the computational complexity of the robustness problem.

We show that RDMA-robustness can be reduced to reachability in a finite-state program augmented with net counters ($\in \Nat$), which implies membership in \expspace:
This follows immediately from the fact that the state reachability problem for Vector Addition Systems with States (VASS) is in \expspace \cite{Rackoff78}.
For RDMA programs that do not use poll operations, the reduction yields a purely finite-state program without counters.
Since state reachability for finite-state programs is in \pspace~\cite{SistlaC85}, we obtain \pspace-membership for the special case.

\begin{theorem}[Upper Bounds for RDMA-Robustness]
\label{thm:complexity_robustness}
    \begin{enumerate}
        \item RDMA-robustness is in \expspace in the general case.
        \item RDMA-robustness for programs without poll operations is in \pspace.
    \end{enumerate}
\end{theorem}

The key to establishing the upper bounds is to construct an instrumented program running under SC that systematically searches for violating traces in normal form.
The instrumented program builds up the two partitions $E_1$ and $E_2$ in parallel, non-deterministically assigning each generated event to either partition.
To implement this effectively, we work with linearisations $\tau_1$ and $\tau_2$ (explicit total orderings of $E_1$ and $E_2$) rather than the partitions directly.
This choice is motivated by the need to verify RDMA-consistency: maintaining explicit orderings allows us to check that all \ob-edges are forward edges, which is simpler than reasoning about acyclicity in a partial order.

In the following, $\tau = \tau_1 \tau_2$ denotes the violating trace in normal form that we are constructing.
Events are appended to the end of either $\tau_1$ or $\tau_2$, preserving SC-consistency within each sequence.
The instrumented program ensures that the concatenation $\tau = \tau_1 \tau_2$ forms an RDMA-consistent trace with an SC-cycle, thereby witnessing non-robustness of the original RDMA program.

\subsection{Key Idea}

The key idea is to simulate the program $\rdmaprog$ and, for each generated event, non-deterministically decide whether it belongs to $\tau_1$ or $\tau_2$.
The instrumented program runs two simulations in parallel, building up $\tau_1$ and $\tau_2$ in the process:
The first simulation tracks the effects of the events of $\tau_1$ on the memory, while the second simulation does the same for the events of $\tau_2$ on a separate copy of the memory.

Crucially, generated events are always appended to the end of either $\tau_1$ or $\tau_2$.
This preserves SC-consistency \emph{within} each $\tau_i$, but not necessarily \emph{between} $\tau_1$ and $\tau_2$.
The instrumented program must keep track of and ensure RDMA-consistency of the resulting trace $\tau = \tau_1 \tau_2$.
Since we are trying to guess $\tau$ as an \ob-linearisation, all \ob-edges must be forward edges.
If all edges are both well-formed and forward edges, then the resulting trace is RDMA-consistent.

\paragraph{Memory consistency.}
The memory after simulating $\tau_1$ must match the memory at the start of the simulation of $\tau_2$.
To ensure this, we non-deterministically guess the memory state after $\tau_1$ at the start of the program and use this for the simulation of $\tau_2$.
After the instrumented program has simulated all of $\tau$, it checks that the memory of the $\tau_1$-simulation matches what the $\tau_2$-simulation started with.

\paragraph{SC-cycle detection.}
The algorithm must also check whether there exists an SC-cycle.
An SC-cycle consists of a cycle of events where each event and its successor are related by \po, \rf, \rb, or \mo.
At a high level, we guess which threads participate in the cycle and in which order they are visited.
During simulation, each thread may mark one entry and one exit event as its local contribution to the cycle and store the corresponding metadata (type, variable, value, and whether the event is in $\tau_1$ or $\tau_2$).
After simulation stops, the main thread checks whether these local guesses compose into a valid global SC-cycle.
The precise justification and checks are given in \autoref{sec:sc_cycles}.

\subsection{Algorithm Description}

The instrumented program consists of one main thread that orchestrates the simulation plus one simulation thread for each thread in the original RDMA program.
Additionally to some bookkeeping and status variables, the instrumented program stores the following:
\begin{itemize}
    \item Three copies of each variable of the RDMA program: one copy to simulate $\tau_1$ and $\tau_2$ respectively, and the third to check that the value after $\tau_1$ matches the memory at the beginning of $\tau_2$.
    \item For each thread, information about the entry and exit point of the SC cycle, together with the threads's predecessor and successor in the SC cycle.
    \item For each queue pair, information about which type of events have already been added to $\tau_2$.
    \item For each queue pair, two counters to track the number of polls that can be executed in $\tau_1$ and $\tau_2$, respectively.
\end{itemize}
The instrumented program works as follows:
\begin{enumerate}
    \item
    The main thread guesses the value of each variable after executing $\tau_1$.
    It keeps each value in two copies of the variable:
    one for the $\tau_2$-simulation to work on, and one to later verify that the value matches the result after $\tau_1$.
    \item
    The main thread signals the other threads to start simulating.
    \item
    Concurrently, each simulation thread of the instrumented program simulates a thread of the RDMA program as follows.
    For each simulated event:
    \begin{enumerate}
        \item Non-deterministically decide whether the event should be in $\tau_1$ or $\tau_2$.
        \item Check that the event can be added to $\tau_1$ or $\tau_2$, respectively, without violating RDMA-consistency (see \autoref{sec:rdma_consistency}).\label{item:rdma_consistency_check}
        \item Execute the effect of the event on the copy of the memory corresponding to $\tau_1$ or $\tau_2$, respectively.
        \item
        Non-deterministically decide whether the event should be the entry point of the SC-cycle in this thread.
        If so, save the event and check that it can form a valid connection (via \rf, \mo, or \rb) with the corresponding exit point of its predecessor thread in the SC-cycle (see \autoref{sec:sc_cycles}).
        \item Do the same for the exit point.
    \end{enumerate}
    \item
    Non-deterministically, the main thread decides to signal the other threads to stop the simulation.
    It waits for each thread to finish their current event.
    \item
    The main thread checks that the memory after $\tau_1$ matches the memory at the beginning of $\tau_2$.
    \item
    The main thread checks that there is a valid SC-cycle.
    If so, it enters the \emph{target state}.
\end{enumerate}
The original RDMA program is robust if and only if the instrumented SC program \emph{cannot} reach the target state.

\subsection{Ensuring RDMA-Consistency}
\label{sec:rdma_consistency}

To ensure that an event can be added to the trace in an RDMA-consistent way (step~\ref{item:rdma_consistency_check}), the program must check that each \ob-edge is a forward edge.
Additionally, it must verify that \rf is well-formed, i.e., that the value of a read event matches the value of its write event.

The most challenging cases are \oppo and $([\Inst];\ib)$, which can potentially form backward edges.
For \oppo, we track per queue pair which event types have been added to $\tau_2$, preventing events from being added to $\tau_1$ if this would create a backward edge.
For $([\Inst];\ib)$, we perform additional checks based on case analysis to ensure all paths remain forward-directed.
The poll-from relation $([\nlW];\pf)$ requires counters to track unpolled NIC writes, with additional bookkeeping to force certain polls into $\tau_2$.
The remaining relations (\rf, \nfo, \rb, \mo) are always forward edges and require only basic well-formedness checks.

Details of all consistency checks can be found in Appendix~\ref{sec:rdma_consistency_details}.

\subsection{Detecting SC-Cycles}
\label{sec:sc_cycles}

To detect SC-cycles, we non-deterministically guess the entry and exit points of the SC-cycle in each thread.
The entry and exit points are the \po-first and \po-last events of the thread on the cycle, respectively.

A key observation is that it is sufficient to consider cycles that visit each thread at most once.
This is independent of the normal-form assumption (and of RDMA robustness): it holds for any cycle over \po, \rf, \rb, and \mo edges.
If a cycle visits some thread more than once, let $e_1$ be the \po-latest cycle event on that thread, and let $e_2$ be the next cycle event on the same thread encountered when following the cycle after $e_1$.
Then the segment from $e_2$ to $e_1$ can be contracted to the single edge $e_2 \to[\po] e_1$, yielding a shorter cycle.
Repeating this contraction removes repeated visits, so an equivalent cycle exists that visits each thread at most once.

Therefore, it is enough to check connectivity between consecutive threads in the guessed cycle order by verifying that their selected exit/entry events are related by \rf, \mo, or \rb.

The instrumented program ensures proper ordering between exit and entry points across partition boundaries ($\tau_1$ and $\tau_2$) through appropriate tracking during simulation.
After all threads finish, the main thread verifies that all pairs of entry and exit points form a valid cycle.

Details can be found in Appendix~\ref{sec:sc_cycle_detection_details}.

\subsection{Complexity Analysis}

The instrumented program described above establishes the upper bounds of robustness as in \autoref{thm:complexity_robustness}.
It uses a finite amount of state to track:
\begin{itemize}
    \item Three copies of the memory (for $\tau_1$ and $\tau_2$ simulations and comparison),
    \item Per-queue-pair tracking of event types added to $\tau_2$ (for \oppo checking),
    \item Information about entry and exit points of the SC-cycle,
    \item Additional bookkeeping for \ib-consistency checks.
\end{itemize}
All of this requires space polynomial in the size of the original RDMA program.

For programs with poll operations, the instrumented program additionally uses two counters per queue pair to track unpolled NIC writes.
This makes the instrumented program a finite-state program augmented with net counters, which is essentially a vector addition system with states.
Since the reachability problem for VASS is decidable in \expspace~\cite{Rackoff78}, we obtain \expspace-membership.

For programs without poll operations, no counters are needed, yielding a purely finite-state instrumented program.
Reachability for finite-state programs is in \pspace~\cite{SistlaC85}, thus robustness for poll-free programs is in \pspace.

To validate this construction, we have implemented a tool that translates RDMA programs to instrumented Promela programs, which can then be verified using the SPIN model checker.
SPIN is not intended for unbounded VASS-style counter reachability in full generality; however, it is sufficient for our proof-of-concept experiments on small acyclic litmus tests, where the number of polls is bounded and the translation therefore requires only bounded counters.
We evaluated our tool on all sequential and concurrent litmus tests created by Ambal et al~\cite{AmbalDEK0R24}, all of which produced the expected results.
\section{Robustness is EXPSPACE-Complete}
\label{sec:robust-hardness}

The previous section established that RDMA-robustness is decidable in \expspace for the general case and in \pspace for programs without poll operations.
We now show that these upper bounds are tight by establishing matching lower bounds.
This demonstrates that the algorithm from the previous section is optimal.

We establish \expspace-hardness by reduction from the state-reachability problem for Vector Addition Systems with States (VASS)~\cite{HopcroftP79}.
A VASS consists of a finite-state control with a finite number of counters ranging over the natural numbers $\Nat$.
Each transition either increments or decrements a counter.
Increment transitions are always enabled, while decrement transitions can only fire when the counter value is strictly positive.
The state-reachability problem asks whether, starting from a given initial state with all counters set to zero, there exists a sequence of transitions leading to a designated accepting state (with arbitrary counter values).
This problem is \expspace-complete~\cite{Lipton76,Rackoff78}.

The key insight is that RDMA programs can naturally simulate VASS executions: control states are encoded in program variables, while counter values correspond to the number of unpolled remote writes on queue pairs.
We construct an RDMA program that simulates a VASS and then extend it to witness non-robustness precisely when the VASS reaches its accepting state.

\subsection{Reduction from VASS}

The reduction simulates a VASS execution using an RDMA program $\rdmaprog$.
Given a VASS $V=(Q,d,T,q_0)$ with control states $Q$, $d$ counters, transitions $T$, and initial state $q_0$, we construct an RDMA program as follows.
For each counter $c \in \{1,\dots,d\}$, we introduce an auxiliary node $n_c$ containing a single variable $x_c$.
A distinguished leader node $n_L$ simulates the control state of the VASS using a local variable $x$.

The reduction maintains the following invariants:
The leader node stores the current VASS control state in $x$.
For each counter $c$, its value is represented by the number of remote writes to node $n_c$ that have not yet been matched by a poll operation.
Thus, issuing a remote write increments the counter, while a poll operation decrements it.
Since each poll operation matches exactly one pending write, the number of pending acknowledgements corresponds precisely to the counter value.
\begin{wrapfigure}[11]{r}{0.6\textwidth}
    \centering
    \begin{subfigure}{0.25\textwidth}
        \centering
        \begin{tabular}{l l}
            \texttt{1:} & $\assume(x = q)$ \\
            \texttt{2:} & $\overline x_{c} := 1$ \\
            \texttt{3:} & $x := q'$
        \end{tabular}
    \end{subfigure}
    \hfill
    \begin{subfigure}{0.25\textwidth}
        \centering
        \begin{tabular}{l l}
            \texttt{1:} & $\assume(x = q)$ \\
            \texttt{2:} & $\poll (n_c)$ \\
            \texttt{3:} & $x := q'$
        \end{tabular}
    \end{subfigure}
    \caption{The RDMA code fragment corresponding to VASS transitions incrementing (left) and decrementing (right) a counter $c$.}
    \label{fig:vass-rdma}
\end{wrapfigure}
More precisely, for each transition $t=(q,c,j,q') \in T$, the leader uses the code fragment given in \autoref{fig:vass-rdma} to simulate the transition $t$.
It first checks that the value of $x$ is $q$, then performs the RDMA operations corresponding to the value of $j$: an increment of counter $c$ (i.e., $j=1$) is simulated by issuing a remote write to $x_c$, while a decrement (i.e., $j=-1$) is simulated by a poll operation on $n_c$.

Suppose the VASS has $k$ transitions $\{t_1, t_2, \dots, t_k\}$.
The leader repeatedly executes a non-deterministic choice among all such code fragments (one per transition), thereby simulating arbitrary VASS transition sequences.
Formally, let $\C_i$ denote the code fragment corresponding to transition $i$.
Let $\C_a$ be a code fragment that checks whether $x$ equals the accepting state $q_a$.
The leader node executes the following loop: $(\C_1 + \C_2 + \dots + \C_k + \C_a)^*$.
It follows that the accepting state $q_a$ is reachable in the VASS if and only if $\C_a$ is reachable in $\rdmaprog$.
The construction is polynomial in the size of the VASS.
Note that $\rdmaprog$ is robust: all operations execute in program order from the leader's perspective, and even though NIC writes may be delayed past subsequent polls, this is not observable since there are no competing reads or writes to the auxiliary nodes.

\subsection{Extension to Robustness}

To reduce to the robustness problem, we modify the RDMA program $\rdmaprog$ into a new program $\rdmaprog'$ such that $\rdmaprog'$ is non-robust if and only if the state $q_a$ is reachable in the VASS.
We append some code fragment known to be non-robust (for example, the code from \autoref{sec:example}) to the end of $\C_a$ to obtain the new code $\C'_a$.
If needed, we also add additional nodes and variables to the RDMA program, as required by the non-robust code fragment.
Let $\rdmaprog'$ be the program $(\C_1 + \C_2 + \dots + \C_k + \C'_a)^*$.

If $q_a$ is reachable in the VASS (i.e., $\C_a$ is reachable in $\rdmaprog$), then $\C'_a$ is reachable in $\rdmaprog'$, causing a robustness violation.
Conversely, if $q_a$ is not reachable in the VASS, then the code fragment $\C'_a$ can never be executed, and hence there is no robustness violation.
This establishes \expspace-hardness for the general case.

For the restricted case of RDMA programs without poll operations, we use a simpler construction.
Since no counters are needed, we reduce from the reachability problem for ordinary finite-state programs instead of VASS, which is \pspace-hard~\cite{CHENG1995117}.

\begin{theorem}[Complexity of RDMA-Robustness]
\label{thm:robustness-completeness}
    \begin{enumerate}
        \item RDMA-robustness is \expspace-complete in the general case.
        \item RDMA-robustness for programs without poll operations is \pspace-complete.
    \end{enumerate}
\end{theorem}

This establishes that the algorithm presented in the previous section is optimal: the \expspace upper bound cannot be improved without advances in the complexity of VASS reachability, and similarly for the \pspace bound in the poll-free case.
The robustness problem is thus computationally equivalent to reachability in finite-state programs with counters (general case) or without counters (poll-free case).

\section{Conclusion}\label{sec:conclusion}

In this paper, we studied verification problems for RDMA programs and showed that reachability is undecidable in general, even for restricted fragments.
We then proved that robustness for RDMA programs is decidable and \expspace-complete, based on a normal-form characterisation of robustness violations.
This result yields a decision procedure that reduces robustness checking to reachability under sequential consistency, with a \pspace bound for programs without poll operations.
In contrast to prior RDMA robustness work that provided restrictive sufficient criteria, our results give full decidability and complexity characterisations for robustness and for reachability-based safety verification under RDMA semantics.

We hope that our work will serve as a foundation for practical tools for automated robustness checking of RDMA programs.
Our own implementation is a proof of concept, and scaling it to real-world benchmarks remains a challenge.
We also plan to extend the techniques to richer RDMA features and to other baseline memory models.

\begin{credits}
\subsubsection{\discintname}
The authors have no competing interests to declare.
\end{credits}

\bibliography{rdma}

@article{CHENG1995117,
  author       = {Allan Cheng and
                  Javier Esparza and
                  Jens Palsberg},
  title        = {Complexity Results for 1-Safe Nets},
  journal      = {Theor. Comput. Sci.},
  volume       = {147},
  number       = {1{\&}2},
  pages        = {117--136},
  year         = {1995},
  doi          = {https://doi.org/10.1016/0304-3975(94)00231-7},
}

@inproceedings{DBLP:conf/icalp/BouajjaniMM11,
  author       = {Ahmed Bouajjani and
                  Roland Meyer and
                  Eike M{\"{o}}hlmann},
  editor       = {Luca Aceto and
                  Monika Henzinger and
                  Jir{\'{\i}} Sgall},
  title        = {Deciding Robustness against Total Store Ordering},
  booktitle    = {Automata, Languages and Programming - 38th International Colloquium,
                  {ICALP} 2011, Zurich, Switzerland, July 4-8, 2011, Proceedings, Part
                  {II}},
  series       = {Lecture Notes in Computer Science},
  volume       = {6756},
  pages        = {428--440},
  publisher    = {Springer},
  year         = {2011},
  doi          = {https://doi.org/10.1007/978-3-642-22012-8_34},
}

@article{AmbalDEK0R24,
  author       = {Guillaume Ambal and
                  Brijesh Dongol and
                  Haggai Eran and
                  Vasileios Klimis and
                  Ori Lahav and
                  Azalea Raad},
  title        = {Semantics of Remote Direct Memory Access: Operational and Declarative
                  Models of {RDMA} on {TSO} Architectures},
  journal      = {Proc. {ACM} Program. Lang.},
  volume       = {8},
  number       = {{OOPSLA2}},
  pages        = {1982--2009},
  year         = {2024},
  doi          = {https://doi.org/10.1145/3689781},
}

@inproceedings{DBLP:conf/esop/AmbalLR25,
  author       = {Guillaume Ambal and
                  Ori Lahav and
                  Azalea Raad},
  title        = {Sufficient Conditions for Robustness of {RDMA} Programs},
  booktitle    = {{ESOP} {(1)}},
  series       = {Lecture Notes in Computer Science},
  volume       = {15694},
  pages        = {56--87},
  publisher    = {Springer},
  year         = {2025},
  doi          = {https://doi.org/10.1007/978-3-031-91118-7_3},
}

@inproceedings{DanLHV16,
  author       = {Andrei Marian Dan and
                  Patrick Lam and
                  Torsten Hoefler and
                  Martin T. Vechev},
  title        = {Modeling and analysis of remote memory access programming},
  booktitle    = {{OOPSLA}},
  pages        = {129--144},
  publisher    = {{ACM}},
  year         = {2016},
  doi          = {https://doi.org/10.1145/2983990.2984033},
}

@article{Post1946,
  author       = {Emil L. Post},
  title        = {A Variant of a Recursively Unsolvable Problem},
  journal      = {Bulletin of the American Mathematical Society},
  volume       = {52},
  number       = {4},
  pages        = {264--268},
  year         = {1946},
  doi          = {https://doi.org/10.1090/S0002-9904-1946-08555-9},
}

@article{HopcroftP79,
  author       = {John E. Hopcroft and
                  Jean{-}Jacques Pansiot},
  title        = {On the Reachability Problem for 5-Dimensional Vector Addition Systems},
  journal      = {Theor. Comput. Sci.},
  volume       = {8},
  pages        = {135--159},
  year         = {1979},
  doi          = {https://doi.org/10.1016/0304-3975(79)90041-0},
}

@book{Lipton76,
  title        = {The reachability problem requires exponential space},
  author       = {Richard J. Lipton},
  series       = {Research report (Yale University. Department of Computer Science)},
  url          = {https://books.google.de/books?id=7iSbGwAACAAJ},
  year         = {1976},
  publisher    = {Department of Computer Science, Yale University}
}

@article{Rackoff78,
  author       = {Charles Rackoff},
  title        = {The Covering and Boundedness Problems for Vector Addition Systems},
  journal      = {Theor. Comput. Sci.},
  volume       = {6},
  pages        = {223--231},
  year         = {1978},
  doi          = {https://doi.org/10.1016/0304-3975(78)90036-1},
}

@article{SistlaC85,
  author = {Sistla, A. P. and Clarke, E. M.},
  title = {The complexity of propositional linear temporal logics},
  year = {1985},
  issue_date = {July 1985},
  publisher = {Association for Computing Machinery},
  address = {New York, NY, USA},
  volume = {32},
  number = {3},
  issn = {0004-5411},
  doi = {https://doi.org/10.1145/3828.3837},
  journal = {J. ACM},
  month = jul,
  pages = {733–749},
  numpages = {17}
}

@article{AbdullaABKS24,
  author       = {Parosh Aziz Abdulla and
                  Mohamed Faouzi Atig and
                  Ahmed Bouajjani and
                  K. Narayan Kumar and
                  Prakash Saivasan},
  title        = {Verification under Intel-x86 with Persistency},
  journal      = {Proc. {ACM} Program. Lang.},
  volume       = {8},
  number       = {{PLDI}},
  pages        = {1189--1212},
  year         = {2024},
  doi          = {https://doi.org/10.1145/3656425},
}

@article{AbdullaABKS21,
  author       = {Parosh Aziz Abdulla and
                  Mohamed Faouzi Atig and
                  Ahmed Bouajjani and
                  K. Narayan Kumar and
                  Prakash Saivasan},
  title        = {Deciding reachability under persistent x86-TSO},
  journal      = {Proc. {ACM} Program. Lang.},
  volume       = {5},
  number       = {{POPL}},
  pages        = {1--32},
  year         = {2021},
  doi          = {https://doi.org/10.1145/3434337},
}

@inproceedings{AbdullaAGKV21,
  author       = {Parosh Aziz Abdulla and
                  Mohamed Faouzi Atig and
                  Adwait Godbole and
                  S. Krishna and
                  Viktor Vafeiadis},
  title        = {The Decidability of Verification under {PS} 2.0},
  booktitle    = {{ESOP}},
  series       = {Lecture Notes in Computer Science},
  volume       = {12648},
  pages        = {1--29},
  publisher    = {Springer},
  year         = {2021},
  doi          = {https://doi.org/10.26226/morressier.604907f41a80aac83ca25d26},
}

@inproceedings{AbdullaABDLM20,
  author       = {Parosh Aziz Abdulla and
                  Mohamed Faouzi Atig and
                  Ahmed Bouajjani and
                  Egor Derevenetc and
                  Carl Leonardsson and
                  Roland Meyer},
  title        = {On the State Reachability Problem for Concurrent Programs Under Power},
  booktitle    = {{NETYS}},
  series       = {Lecture Notes in Computer Science},
  volume       = {12129},
  pages        = {47--59},
  publisher    = {Springer},
  year         = {2020},
  doi          = {https://doi.org/10.1007/978-3-030-67087-0_4},
}

@inproceedings{AbdullaABN16,
  author       = {Parosh Aziz Abdulla and
                  Mohamed Faouzi Atig and
                  Ahmed Bouajjani and
                  Tuan Phong Ngo},
  title        = {The Benefits of Duality in Verifying Concurrent Programs under {TSO}},
  booktitle    = {{CONCUR}},
  series       = {LIPIcs},
  volume       = {59},
  pages        = {5:1--5:15},
  publisher    = {Schloss Dagstuhl - Leibniz-Zentrum f{\"{u}}r Informatik},
  year         = {2016},
  doi          = {https://doi.org/10.4230/LIPIcs.CONCUR.2016.5},
}

@inproceedings{AbdullaAP15,
  author       = {Parosh Aziz Abdulla and
                  Mohamed Faouzi Atig and
                  Ngo Tuan Phong},
  title        = {The Best of Both Worlds: Trading Efficiency and Optimality in Fence
                  Insertion for {TSO}},
  booktitle    = {{ESOP}},
  series       = {Lecture Notes in Computer Science},
  volume       = {9032},
  pages        = {308--332},
  publisher    = {Springer},
  year         = {2015},
  doi          = {https://doi.org/10.1007/978-3-662-46669-8_13},
}

@inproceedings{SinghL24,
  author       = {Abhishek Kr Singh and
                  Ori Lahav},
  title        = {Decidable Verification under Localized Release-Acquire Concurrency},
  booktitle    = {{TACAS} {(3)}},
  series       = {Lecture Notes in Computer Science},
  volume       = {14572},
  pages        = {235--254},
  publisher    = {Springer},
  year         = {2024},
  doi          = {https://doi.org/10.1007/978-3-031-57256-2_12},
}

@article{Toplas-LahavB22,
  author       = {Ori Lahav and
                  Udi Boker},
  title        = {What's Decidable About Causally Consistent Shared Memory?},
  journal      = {{ACM} Trans. Program. Lang. Syst.},
  volume       = {44},
  number       = {2},
  pages        = {8:1--8:55},
  year         = {2022},
  doi          = {https://doi.org/10.1145/3505273},
}

@inproceedings{AtigBBM12,
  author       = {Mohamed Faouzi Atig and
                  Ahmed Bouajjani and
                  Sebastian Burckhardt and
                  Madanlal Musuvathi},
  title        = {What's Decidable about Weak Memory Models?},
  booktitle    = {{ESOP}},
  series       = {Lecture Notes in Computer Science},
  volume       = {7211},
  pages        = {26--46},
  publisher    = {Springer},
  year         = {2012},
  doi          = {https://doi.org/10.1007/978-3-642-28869-2_2},
}

@inproceedings{AtigBBM10,
  author       = {Mohamed Faouzi Atig and
                  Ahmed Bouajjani and
                  Sebastian Burckhardt and
                  Madanlal Musuvathi},
  title        = {On the verification problem for weak memory models},
  booktitle    = {{POPL}},
  pages        = {7--18},
  publisher    = {{ACM}},
  year         = {2010},
  doi          = {https://doi.org/10.1145/1706299.1706303},
}

@inproceedings{AbdullaAKR15,
  author       = {Parosh Aziz Abdulla and
                  Mohamed Faouzi Atig and
                  Ahmet Kara and
                  Othmane Rezine},
  title        = {Verification of Buffered Dynamic Register Automata},
  booktitle    = {{NETYS}},
  series       = {Lecture Notes in Computer Science},
  volume       = {9466},
  pages        = {15--31},
  publisher    = {Springer},
  year         = {2015},
  doi          = {https://doi.org/10.1007/978-3-319-26850-7_2},
}

@inproceedings{AbdullaAAK19,
  author       = {Parosh Aziz Abdulla and
                  Jatin Arora and
                  Mohamed Faouzi Atig and
                  Shankara Narayanan Krishna},
  title        = {Verification of programs under the release-acquire semantics},
  booktitle    = {{PLDI}},
  pages        = {1117--1132},
  publisher    = {{ACM}},
  year         = {2019},
  doi          = {https://doi.org/10.1145/3314221.3314649},
}

@inproceedings{BouajjaniDM13,
  author       = {Ahmed Bouajjani and
                  Egor Derevenetc and
                  Roland Meyer},
  title        = {Checking and Enforcing Robustness against {TSO}},
  booktitle    = {{ESOP}},
  series       = {Lecture Notes in Computer Science},
  volume       = {7792},
  pages        = {533--553},
  publisher    = {Springer},
  year         = {2013},
  doi          = {https://doi.org/10.1007/978-3-642-37036-6_29},
}

@inproceedings{CalinDMM13,
  author       = {Georgel Calin and
                  Egor Derevenetc and
                  Rupak Majumdar and
                  Roland Meyer},
  title        = {A Theory of Partitioned Global Address Spaces},
  booktitle    = {{FSTTCS}},
  series       = {LIPIcs},
  volume       = {24},
  pages        = {127--139},
  publisher    = {Schloss Dagstuhl - Leibniz-Zentrum f{\"{u}}r Informatik},
  year         = {2013},
  doi          = {https://doi.org/10.4230/LIPIcs.FSTTCS.2013.127},
}

@inproceedings{DerevenetcM14,
  author       = {Egor Derevenetc and
                  Roland Meyer},
  title        = {Robustness against Power is PSpace-complete},
  booktitle    = {{ICALP} {(2)}},
  series       = {Lecture Notes in Computer Science},
  volume       = {8573},
  pages        = {158--170},
  publisher    = {Springer},
  year         = {2014},
  doi          = {https://doi.org/10.1007/978-3-662-43951-7_14},
}

@inproceedings{BouajjaniDM14,
  author       = {Ahmed Bouajjani and
                  Egor Derevenetc and
                  Roland Meyer},
  title        = {Robustness against Relaxed Memory Models},
  booktitle    = {Software Engineering},
  series       = {{LNI}},
  volume       = {{P-227}},
  pages        = {85--86},
  publisher    = {{GI}},
  year         = {2014},
  url          = {https://dl.gi.de/handle/20.500.12116/30973},
}

@phdthesis{DBLP:phd/dnb/Derevenetc15,
  author       = {Egor Derevenetc},
  title        = {Robustness against Relaxed Memory Models},
  school       = {University of Kaiserslautern},
  year         = {2015},
  url          = {https://nbn-resolving.org/urn:nbn:de:hbz:386-kluedo-40743},
}

\newpage
\appendix
\section{Program Code for the Undecidability Proof}
\label{sec:code}
\setlength{\tabcolsep}{4pt}

This section holds the full program code for each of the processes in the reduction from PCP to RDMA reachability.
We use the expression $\A^+$ to denote $\A;\A^*$, that is, $\A$ is executed for a finite, \emph{non-zero} number of iterations.
Furthermore, $\bar x := c$ for some remote variable $\bar x$ and constant $c$ assumes that the value $c$ is available in some (readonly) variable.

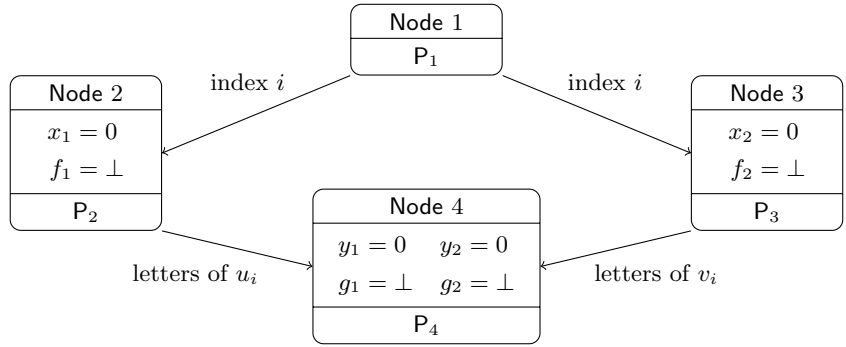
\begin{figure}[H]
    \centering
    \begin{tikzpicture}[
        every text node part/.style={align=center},
        xscale=4.5,
        yscale=-1.5,
        proc/.style={draw, rectangle split, rectangle split parts=3, rounded corners, minimum width=2cm}
    ]
        \node[proc, rectangle split parts=2] (p1) at ( 0, 0) {\text{\nodeset 1} \nodepart{two} $\process_1$};
        \node[proc] (p2) at (-1, 1) {
            \text{\nodeset 2} \nodepart{two}
            $\begin{aligned}x_1 &= 0 \\ f_1 &= \bot\end{aligned}$
            \nodepart{three} $\process_2$
        };
        \node[proc] (p3) at ( 1, 1) {
            \text{\nodeset 3} \nodepart{two}
            $\begin{aligned}x_2 &= 0 \\ f_2 &= \bot\end{aligned}$
            \nodepart{three} $\process_3$
        };
        \node[proc, minimum width=3cm] (p4) at ( 0, 2) {
            \text{\nodeset 4} \nodepart{two}
            $\begin{aligned}y_1 &= 0 & y_2 &= 0 \\ g_1 &= \bot & g_2 &= \bot\end{aligned}$
            \nodepart{three} $\process_4$
        };
        
        \draw[->] (p1.south west) -- node[pos=0.3, above left]  {index $i$}        (p2.east);
        \draw[->] (p1.south east) -- node[pos=0.3, above right] {index $i$}        (p3.west);
        \draw[->] (p2.south east) -- node[pos=0.7, below left]  {letters of $u_i$} (p4.west);
        \draw[->] (p3.south west) -- node[pos=0.7, below right] {letters of $v_i$} (p4.east);
    \end{tikzpicture}
    \caption{Nodes, variables, processes and data flow in the RDMA program.}
    \label{fig:pcp-overview}
\end{figure}

\vspace{-0.6cm}

\begin{figure}[H]
    \centering
    \begin{subfigure}{0.45\textwidth}
        \centering
        \begin{tabular}{l l}
            \texttt{1:} & $(\A_1 + \A_2 + \ldots + \A_N)^+$ \\
            \texttt{2:} & $\bar x_1 := \#$ \\
            \texttt{3:} & $\bar x_2 := \#$
        \end{tabular}
        \caption{main routine}
        \label{lst:p1-main}
    \end{subfigure}
    \hfill
    \begin{subfigure}{0.45\textwidth}
        \centering
        \begin{tabular}{l l}
            \texttt{1:} & $\bar x_1 := i$ \\
            \texttt{2:} & $\bar f_1 := \top$ \\
            \texttt{3:} & $\bar x_2 := i$ \\
            \texttt{4:} & $\bar f_2 := \top$
        \end{tabular}
        \caption{subroutine $\A_i$}
        \label{lst:p1-sub}
    \end{subfigure}
    \caption{Process $\process_1$}
    \label{lst:proc-1}
\end{figure}

\vspace{-0.6cm}

\begin{figure}[H]
    \centering
    \begin{subfigure}{0.45\textwidth}
        \centering
        \begin{tabular}{l l}
            \texttt{ 1:} & $($ \\
            \texttt{ 2:} & $\quad \B_1 + \B_2 + \dots + \B_N$ \\
            \texttt{ 3:} & $\quad x_1 := \bot$ \\
            \texttt{ 4:} & $\quad \assume(f_1 = \top)$ \\
            \texttt{ 5:} & $\quad \assume(f_1 = \bot)$ \\
            \texttt{ 6:} & $\quad f_1 := \top$ \\
            \texttt{ 7:} & $\quad \assume(x_1 = \bot)$ \\
            \texttt{ 8:} & $)^+$ \\
            \texttt{ 9:} & $\assume(x_1 = \#)$ \\
            \texttt{10:} & $\assume(f_1 = \bot)$ \\
            \texttt{11:} & $\bar y_1 := \#$
        \end{tabular}
        \caption{main routine}
        \label{lst:p2-main}
    \end{subfigure}
    \hfill
    \begin{subfigure}{0.45\textwidth}
        \centering
        \begin{tabular}{l l}
            \texttt{1:} & $\assume(x_1 = i)$ \\
            \texttt{2:} & $\bar y_1 := u_{i,1}$ \\
            \texttt{3:} & $\bar g_1 := \top$ \\
            \texttt{4:} & $\bar y_1 := u_{i,2}$ \\
            \texttt{5:} & $\bar g_1 := \top$ \\
            \texttt{\vdots} & \qquad\vdots \\
            \texttt{  } & $\bar y_1 := u_{i,|u_i|}$ \\
            \texttt{  } & $\bar g_1 := \top$
        \end{tabular}
        \caption{subroutine $\B_i$}
        \label{lst:p2-sub}
    \end{subfigure}
    \caption{Process $\process_2$}
    \label{lst:proc-2}
\end{figure}

\begin{figure}[p]
    \centering
    \begin{subfigure}{0.45\textwidth}
        \centering
        \begin{tabular}{l l}
            \texttt{ 1:} & $($ \\
            \texttt{ 2:} & $\quad \C_1 + \C_2 + \dots + \C_N$ \\
            \texttt{ 3:} & $\quad x_2 := \bot$ \\
            \texttt{ 4:} & $\quad \assume(f_2 = \top)$ \\
            \texttt{ 5:} & $\quad \assume(f_2 = \bot)$ \\
            \texttt{ 6:} & $\quad f_2 := \top$ \\
            \texttt{ 7:} & $\quad \assume(x_2 = \bot)$ \\
            \texttt{ 8:} & $)^+$ \\
            \texttt{ 9:} & $\assume(x_2 = \#)$ \\
            \texttt{10:} & $\assume(f_2 = \bot)$ \\
            \texttt{11:} & $\bar y_2 := \#$
        \end{tabular}
        \caption{main routine}
        \label{lst:p3-main}
    \end{subfigure}
    \hfill
    \begin{subfigure}{0.45\textwidth}
        \centering
        \begin{tabular}{l l}
            \texttt{1:} & $\assume(x_2 = i)$ \\
            \texttt{2:} & $\bar y_2 := v_{i,1}$ \\
            \texttt{3:} & $\bar g_2 := \top$ \\
            \texttt{4:} & $\bar y_2 := v_{i,2}$ \\
            \texttt{5:} & $\bar g_2 := \top$ \\
            \texttt{\vdots} & \qquad\vdots \\
            \texttt{  } & $\bar y_2 := v_{i,|v_i|}$ \\
            \texttt{  } & $\bar f_2 := \top$
        \end{tabular}
        \caption{subroutine $\C_i$}
        \label{lst:p3-sub}
    \end{subfigure}
    \caption{Process $\process_3$}
    \label{lst:proc-3}
\end{figure}

\begin{figure}
    \centering
    \begin{subfigure}{0.45\textwidth}
        \centering
        \begin{tabular}{l l}
            \texttt{ 1:} & $($ \\
            \texttt{ 2:} & $\quad \D_{\sigma_1} + \D_{\sigma_2} + \dots + \D_{\sigma_{|\Sigma|}}$ \\
            \texttt{ 3:} & $\quad y_1 := \bot$ \\
            \texttt{ 4:} & $\quad \assume(g_1 = \bot)$ \\
            \texttt{ 5:} & $\quad \assume(g_1 = \top)$ \\
            \texttt{ 6:} & $\quad g_1 := \bot$ \\
            \texttt{ 7:} & $\quad \assume(y_1 = \bot)$ \\
            \texttt{ 8:} & $\quad y_2 := \bot$ \\
            \texttt{ 9:} & $\quad \assume(g_2 = \bot)$ \\
            \texttt{10:} & $\quad \assume(g_2 = \top)$ \\
            \texttt{11:} & $\quad g_2 := \bot$ \\
            \texttt{12:} & $\quad \assume(y_2 = \bot)$ \\
            \texttt{13:} & $)^+$ \\
            \texttt{14:} & $\assume(y_1 = \#)$ \\
            \texttt{15:} & $\assume(g_1 = \bot)$ \\
            \texttt{16:} & $\assume(y_2 = \#)$ \\
            \texttt{17:} & $\assume(g_2 = \bot)$
        \end{tabular}
        \caption{main routine}
        \label{lst:p4-main}
    \end{subfigure}
    \hfill
    \begin{subfigure}{0.45\textwidth}
        \centering
        \begin{tabular}{l l}
            \texttt{1:} & $\assume(y_1 = \sigma)$ \\
            \texttt{2:} & $\assume(y_2 = \sigma)$
        \end{tabular}
        \caption{subroutine $\D_\sigma$}
        \label{lst:p4-sub}
    \end{subfigure}
    \caption{Process $\process_4$}
    \label{lst:proc-4}
\end{figure}

\clearpage
\clearpage
\section{Formal Definition of $\pathset(\process)$}
\label{sec:def_seq}

We inductively define $\pathset(\process) \subseteq \labset^*$ as the set of all finite label sequences that can be generated by the thread $\process$:
\begin{align*}
    \pathset(\nop)                      &= \set\epsilon \\
    \pathset(\process_1 ; \process_2)   &= \set{s_1 \cdot s_2 \mid s_1 \in \pathset(\process_1), s_2 \in \pathset(\process_2)} \\
    \pathset(\process_1 + \process_2)   &= \pathset(\process_1) \cup \pathset(\process_2) \\
    \pathset(\process^0)                &= \set\epsilon \\
    \pathset(\process^{k+1})            &= \pathset(\process^k ; \process) \\
    \pathset(\process^*)                &= \bigcup_{k \geq 0} \pathset(\process^k),\\
    \\
    \pathset(x := v_w)                  &= \set{\lW(x, v_w)} \\
    \pathset(\assume(x = v_r))          &= \set{\lR(x, v_r)} \\
    \pathset(\assume(x \neq v_r))       &= \set{\lR(x, v) \mid v \in \valset \setminus \set{v_r}} \\
    \pathset(x := \CAS(z, v_r, v_w))    &= \set{\CAS(z, v_r, v_w) \cdot \lW(x, v_r)} \\
                                        &\qquad\qquad \cup \set{\lR(z, v) \cdot \lW(x, v) \mid v \in \valset \setminus \set{v_r}} \\
    \\
    \pathset(x := \bar y)               &= \set{\nrR(\bar y, v) \cdot \nlW(x, v, n(\bar y)) \mid v \in \valset} \\
    \pathset(\bar y := x)               &= \set{\nlR(x, v, n(\bar y)) \cdot \nrW(\bar y, v) \mid v \in \valset} \\
    \pathset(\rfence(\bar n))           &= \set{\nF(\bar n)} \\
    \pathset(\poll(\bar n))             &= \set{\poll(\bar n)}
\end{align*}

\clearpage
\section{Proof of \autoref{thm:normal_form}}
\label{sec:proofs}

\paragraph{Notation.}
In this section, we use $T$, $T' = T(\tau')$ and $T'' = T(\tau'')$ as defined previously.
Furthermore, for $\rel \in \set{\po, \pf, \rf, \mo, \nfo, \rb, \ippo, \oppo, \ib, \ob, \sc}$, we write $\rel$, $\rel'$ and $\rel''$ for $T.\rel$, $T'.\rel$ and $T''.\rel$, respectively.

Let $\tau = e_1, \dots, e_m$.
We formalise the event set of $T'$ as $E' = \set{ e_1', \dots, e_m' }$, where $e_i'$ is inductively defined as:
\begin{itemize}
    \item If $e_i = \tuple{\iota, t, \nlR(x, v_r, \bar n)}$,
    let $k = \max\set{ j \mid j < i, e_j \in E.\writes, \loc(e_j) = x }$
    and define $e_i' = \tuple{\iota, t, \nlR(x, v_w(e_k'), \bar n)}$.
    \item Similarly, if $e_i = \tuple{\iota, t, \nrR(\bar y, v_r)}$,
    let $k = \max\set{ j \mid j < i, e_j \in E.\writes, \loc(e_j) = y }$
    and define $e_i' = \tuple{\iota, t, \nrR(x, v_w(e_k'), \bar n)}$.
    \item If $e_i = \tuple{\iota, t, \nrW(\bar y, v_w)}$,
    let $e_k$ be the unique event such that $e_k \to[\po] e_i$ and there is no $e_j$ such that $e_k \to[\po] e_j \to[\po] e_i$.
    Due to the construction of \po from the event sequences generated by the thread $\process_t$, it must be the case that $e_k$ is the matching NIC local read to the NIC remote write $e_i$.
    Since $e_k \to[\ib] e_i$, we have $e_k <_\tau e_i$ and thus $k < i$.
    Define $e_i' = \tuple{\iota, t, \nrW(\bar y, v_r(e_k'))}$.
    \item Similarly, if $e_i = \tuple{\iota, t, \nlW(x, v_w, \bar n)}$,
    let $e_k$ be the unique event such that $e_k \to[\po] e_i$ and there is no $e_j$ such that $e_k \to[\po] e_j \to[\po] e_i$.
    Again, $e_k$ must be the matching event of $e_i$.
    We conclude that $e_k$ is a NIC remote read and furthermore that $k < i$.
    Define $e_i' = \tuple{\iota, t, \nlW(x, v_r(e_k'), \bar n)}$.
    \item Otherwise, $e_i' = e_i$.
\end{itemize}
For simplicity, we will sometimes treat $E'$ as a subset of $E$, although this is not strictly the case.
In particular, we associate any event $\bar{e} \in E \setminus \set{e}$ with the (unique) event $\bar{e}' \in E'$ that agrees on both the program counter and the thread id:
$\iota(\bar{e}) = \iota(\bar{e}')$ and $t(\bar{e}) = t(\bar{e}')$.
We will use $a, b, w, r, \dots$ for events in $E$ and $a', b', w', r', \dots$ for their corresponding events in $E'$.
Note that $T$ and $T''$ have \emph{exactly} the same events, that is, $E = E''$ with no caveats.

For completeness, we also formalise the notation $\tau_i\downarrow_\sigma$.
Let $\sigma = e_1, \dots, e_m$, then it is defined as the subsequence of $\sigma$ containing exactly the elements of $\tau_i$:
$\tau_i\downarrow_\sigma = e_{j_1}, \dots, e_{j_{|\tau_i|}}$, where $j_1 < \dots < j_{|\tau_i|}$ and $e_{j_k} \in \tau_i$ for all $k = 1, \dots, |\tau_i|$.

For a relation $\rel \subset E \times E$ and a total order $\tau$ on $E$, we say that $a \to[\rel] b$ is a \emph{forward edge} if $a <_\tau b$, and a \emph{backward edge} if $b <_\tau a$.

\paragraph{Overview.}
The proof proceeds in two phases, corresponding to the two trace transformations from $T$ to $T'$ and from $T'$ to $T''$, respectively.

\autoref{claim:changed_value_is_ib} shows that any event whose value changes must have been issued after $e$, which is crucial for maintaining structural properties.
\autoref{lem:T_prime_trace} proves that $T'$ is a valid trace by verifying the event set, the polls-from relation and the reads-from relation remain well-formed after the transformation.
\autoref{thm:T_prime_consistent} then establishes RDMA-consistency of $T'$ by showing that the observed-before relation $\ob'$ remains acyclic, primarily by demonstrating that $\ob' \subseteq \ob$.

During the transformation from $T$ to $T'$, the relations \po, \ippo, \oppo, \pf, and \nfo are simply restricted to $E'$, i.e., $\rel' = \rel|_{E'}$ for these relations.
However, \rf and \mo may change more substantially: if $e$ was a write event in a read-from or modification-order chain, the corresponding chains in $T'$ may be modified (e.g., $w \to[\mo] e \to[\rf] r$ in $T$ becomes $w \to[\rf'] r$ in $T'$).
Furthermore, the sequential-consistency relation \sc changes as $e$ is removed from the trace.

We proceed with the analysis of $T''$.
The key structural observation is \autoref{lem:rel_equals_rel_prime_prime}, which establishes that several fundamental relations (\po, \ippo, \oppo, \pf, \rf, \mo, \rb) remain unchanged between $T$ and $T''$.
Building on this, \autoref{thm:T_prime_prime_consistent} proves that $T''$ is RDMA-consistent by showing all $\ob''$-edges are forward edges with respect to $\tau''$.
The proof carefully analyses three positional cases for events in $\tau'' = \tau_1\downarrow_\sigma e \tau_2\downarrow_\sigma$ and shows that the problematic backward edge case leads to a contradiction.
Finally, \autoref{cor:T_prime_prime_violating} establishes that $T.\sc = T''.\sc$, completing the proof that $T''$ violates the same SC properties as $T$ while being RDMA-consistent.

\begin{claim}
\label{claim:changed_value_is_ib}
    If $\bar e \in E$ and its corresponding event $\bar e' \in E'$ have different values, then $\bar e$ must have been issued after the removed event $e$, that is, $e \to[\ib] \bar e$.
\end{claim}
\begin{proof}
    We show that the claim holds for the $k$-th event of $\tau$ by induction on $k$.
    If $k=1$, $\bar e$ cannot be a NIC read or write, since both require an initialisation event beforehand.
    In the case $k > 1$ we assume that the claim holds for all events in $\tau$ up to the $(k-1)$-th event.
    Otherwise, if $\bar e$ is a NIC read event, then there must be a write event $w$ such that $w \to[\rf] \bar e$.
    In the case $w = e$ we are already done, since $\rf \subset \ib$.
    In the general case, let $w'$ be the corresponding event of $w$ in $E'$.
    Since $w' \to[\rf'] \bar e'$ and $v_w(w) = v_r(\bar e) \neq v_r(\bar e') = v_w(w')$, we can apply the induction hypothesis to $w$ and obtain a path $e \to[\ib] w \to[\rf \subset \ib] \bar e$.
    Lastly, if $\bar e$ is a NIC write event, there must be a matching read event $r$ that immediately precedes $\bar e$ in $\po$.
    We apply the same logic as in the previous case to obtain $e \to[\ib] r \to[\ippo \subset \ib] \bar e$.
    Note that $r \to[\po] \bar e$ is in $\ippo$ since $r$ and $\bar e$ are the two matching events of the same remote write or remote read instruction.
\end{proof}

\begin{lemma}
\label{lem:T_prime_trace}
    $T' = \tuple{ E', \po', \pf', \rf', \mo', \nfo' }$ is a trace of \rdmaprog.
\end{lemma}
\begin{proof}
    We need to show that the event set $E'$, the polls-from relation $\pf'$ and the reads-from relation $\rf'$ are well-formed.

    The event set $E'$ was obtained from $E$ by first removing $e$ and then adjusting the values of depending NIC reads and writes.
    Let $t \in \Tid$ be such that $e \in E^t$.
    Due to maximality of $e$ we have $\iota \of e = |E^t|$ and thus $E^t \setminus \set e = \set{\tuple{\iota, t, s_t[\iota]} \mid 1 \leq \iota \leq |E^t| - 1}$ is well-formed.
    For the second step, consider a pair of NIC reads and writes $r, w \in E^t, t\in \Tid$, corresponding to a remote write or read instruction.
    Let $r', w'$ be the corresponding events in $E'$.
    Note that $v_r(r') = v_w(w')$ by construction.
    By the definition of $\pathset$, replacing $r$ and $w$ in $s_t$ with $r'$ and $w'$ yields another sequence $s'_t \in \pathset(\process_t)$.
    Thus, changing the values of pairs of NIC reads and writes preserves the well-formedness of the event set.

    For \pf', we only need to show that $E'$ has at least as many NIC writes as polls on each queue pair.
    Since $E' = E \setminus \set{e}$, this can only be violated if $e$ is a polled NIC write.
    However, this contradicts the maximality of $e$, since any poll event that polls from $e$ would be \po-later than $e$.

    For \rf', consider $w' \to[\rf'] r'$.
    We need to show $v_w(w') = v_r(r')$, the rest is straightforward.
    If $r'$ is a NIC (either local or remote) read, then its value was defined to match $w'$.
    Otherwise, $r'$ is either a successful CAS event or a local read, both of which are CPU instructions.
    Let $r$ be the corresponding event of $r'$ in $E$ and let $w$ be such that $w \to[\rf] r$.
    We can rule out $w = e$, since this would imply $e \to[\rf \subset \ib] r$ which contradicts maximality of $e$.
    It follows that $w \neq e$ is the corresponding event of $w'$ in $E$.
    If both $w'$ and $r'$ have the same value as $w$ and $r$ respectively, then we are done.
    Otherwise, by \autoref{claim:changed_value_is_ib}, we have either $e \to[\ib] w \to[\rf \subset \ib] r$ or $e \to[\ib] r$, which both contradict maximality of $e$.
\end{proof}

\begin{theorem}
\label{thm:T_prime_consistent}
    $T'$ is an RDMA-consistent trace of \rdmaprog.
\end{theorem}
\begin{proof}
    \autoref{lem:T_prime_trace} shows that $T'$ is a trace of \rdmaprog.
    For consistency, we need to show that $\ob'$ is irreflexive.
    In particular, we show that $\ob' \subseteq \ob$.
    Consider $a' \to[\rel'] b'$ and let $a$ and $b$ be the corresponding events in $E$.
    For each $\rel \in \set{\oppo, \rf, ([\nlW]; \pf), \nfo, \rb, \mo, ([\Inst]; \ib)}$ we show that $a \to[\ob] b$:
    \begin{itemize}
        \item
        Since $T = T(\tau)$ and $T' = T(\tau')$, it follows by construction that $\rel' = \rel|_{E'} \subset \rel$ for $\rel \in \set{\oppo, \nfo, \mo}$.
        The same holds for $([\nlW]; \pf)$, since due to the maximality of $e$, it cannot be the case that $e$ is a NIC write that is polled from.

        \item
        If $\rel = \rf$, then either $a \to[\rf] b$, or $a \to[\mo] e \to[\rf] b$ and $e$ is a write event that immediately follows $a$ in $\mo$.
        In both cases it holds that $a \to[\ob] b$.

        \item
        If $\rel = \rb$, then $a' \to[\rb'] b'$ is due to a write event $w'$ where $w' \to[\rf'] a'$ and $w' \to[\mo'] b'$.
        Let $w$ be the write event such that $w \to[\rf] a$ (either $w$ is the corresponding event of $w'$ or $w = e$).
        Since $a < b$ with respect to both $\tau$ and $\tau'$, it holds in particular that $w <_\tau b$.
        Thus, $w \to[\mo] b$ and therefore $a \to[\rb] b$.

        \item
        In the case of $\rel = ([\Inst]; \ib)$, we first observe that similar to above, $\rel' = \rel|_{E'} \subset \rel$ for $\rel \in \set{\ippo, \pf, \nfo}$.
        Let $a' = e_0', \dots, e_n' = b'$ be an $\ib'$-path from $a'$ to $b'$, that is, $e_i' \to[\rel_i'] e_{i+1}'$ and $\rel_i \in \set{ \ippo, \rf, \pf, \nfo}$ for all $0 \leq i < n$.
        Furthermore, let $a = e_0, \dots, e_n = b$ be the corresponding events in $E$.
        If there are no events $w$ and $r$ such that $w \to[\mo] e \to[\rf] r$ as in the second case above, then $e_i \to[\ib] e_{i+1}$ for all $0 \leq i < n$ and we are done.
        Otherwise, choose $r$ to be the last such event.
        We then have the following path:
        $$ a \to[\ib] w \to[\mo] e \to[\rf] r \to[\ib] b $$
        Each of these edges is in \ob, because both $a$ and $r$ are instantaneous.
    \end{itemize}
\end{proof}

\begin{lemma}
\label{lem:rel_equals_rel_prime_prime}
    For $\rel \in \set{\po, \ippo, \oppo, \pf, \rf, \mo, \rb}$, it holds that $\rel'' = \rel$.
\end{lemma}
\begin{proof}
    The claim holds trivially for \po, \ippo, \oppo and \pf, since they depend statically on the event set $E'' = E$.
    For $\rel \in \set{\rf, \mo}$, consider $a \to[\rel] b$.
    It holds that $a <_\tau b$, since \rf and \mo are subsets of \ob and $\tau$ is a linearisation of \ob.
    If $a \in \tau_1\downarrow_\sigma e$ and $b \in e \tau_2\downarrow_\sigma$, then $a <_{\tau''} b$ follows trivially.
    Otherwise $a,b \in \tau_j\downarrow_\sigma$ for some $j \in \set{1,2}$.
    In that case, we observe that $a <_{\tau'} b$ by construction of $\tau'$.
    Since $a$ and $b$ must be related by $\rel'$ (note that for $\rel = \rf$ there cannot be a write event on the same variable between $a$ and $b$), it even holds that $a \to[\rel'] b$.
    Now, $T(\tau').\sc = T(\sigma).\sc$ implies $a <_\sigma b$, which furthermore implies $a <_{\tau''} b$.

    If $\rel = \mo$, we can already conclude that $a \to[\mo''] b$, since $a$ and $b$ must be $\mo''$-related and $a <_{\tau''} b$.
    For $\rel = \rf$, we need to additionally show that there is no other write event on the same variable (\emph{conflicting event}) between $a$ and $b$ in $\tau''$.
    Let $w$ be a write on the same variable occuring after $b$ with respect to $\tau$, that is, $a \to[\mo] w$ and $b <_\tau w$.
    Since we already know $a <_{\tau'} b$ and $\mo'' = \mo$, all we need to show is that $b <_{\tau''} w$.
    If $b = e$ or $w = e$, this is clearly the case by the construction of $\tau''$.
    Otherwise, we have $b <_{\tau'} w$ and therefore $b \to[\rb'] w$.
    Since $T(\tau').\sc = T(\sigma).\sc$, it follows that $b \to[T(\sigma).\rb] w$ and particularly $b <_{\sigma} w$.
    This implies $b <_{\tau''} w$, since $w \in \tau_1$, $b \in \tau_2$ is not possible due to $b <_\tau w$.

    Lastly, for $\rel = \rb$ the statement follows since \rb is defined in terms of \rf and \mo.
\end{proof}

\begin{theorem}
\label{thm:T_prime_prime_consistent}
    $T'' = \tuple{ E'', \po'', \pf'', \rf'', \mo'', \nfo'' }$ is an RDMA-consistent trace of \rdmaprog.
\end{theorem}
\begin{proof}
    The event set $E''$ is well-formed since it is the same set as the event set $E$ of the trace $T$.
    $\rf''$ is well-formed since it coincides with $\rel$ by \autoref{lem:rel_equals_rel_prime_prime}.
    Thus, $T''$ is a trace of \rdmaprog.

    For consistency, we will show that $\ob''$ has only forward edges with respect to $\tau''$ and is therefore acyclic.
    The relations \rf, \nfo, \mo and \rb have only forward edges by definition of $T(\tau'')$.
    For $\rel \in \set{\oppo, ([\nlW]; \pf), ([\Inst]; \ib)}$, consider $a \to[\rel''] b$.
    We distinguish cases based on the position of $a$ and $b$ within $\tau'' = \tau_1\downarrow_\sigma e \tau_2\downarrow_\sigma$.
    First, if $a \in \tau_1\downarrow_\sigma e$ and $b \in e \tau_2\downarrow_\sigma$, then $a \to b$ is trivially a forward edge.
    Second, if $a,b \in \tau_j\downarrow_\sigma$, $j \in \set{1,2}$, then it is a forward edge since the events within each $\tau_j$ follow \sc order, which includes \po and \rf and therefore also \oppo, \ippo and \pf.
    The last case is $a \in e \tau_2\downarrow_\sigma$ and $b \in \tau_1\downarrow_\sigma e$, which clearly implies a backward edge.
    The remainder of the proof is dedicated to show that this case cannot occur.

    Since $\oppo'' = \oppo$ and $\pf'' = \pf$, if $\rel \in \set{\oppo, ([\nlW]; \pf)}$ then this would imply the backward \ob-edge $a \to[\rel] b$ with respect to $\tau$.
    Thus, we are left with the case $\rel = ([\Inst]; \ib)$.
    Let $a = e_0, \dots, e_n = b$ be an $\ib''$-path from $a$ to $b$, that is, $e_i \to[\rel_i''] e_{i+1}$ and $\rel_i \in \set{ \ippo, \rf, \pf, \nfo}$ for all $0 \leq i < n$.
    We consider the position of $a$ and $b$ within $\tau'' = \tau_1\downarrow_\sigma e \tau_2\downarrow_\sigma$.
    If $a \in \tau_1\downarrow_\sigma e$ and $b \in e \tau_2\downarrow_\sigma$, then $a \to b$ is trivially a forward edge.
    If $a,b \in \tau_j\downarrow_\sigma$, $j \in \set{1,2}$, then each $e_i \to[\rel_i''] e_{i+1}$ is a forward edge:
    \rf and \nfo are always forward edges, and if $\rel_i \in \set{ \ippo, \pf}$, then it is a forward edge since the events within each $\tau_j$ follow \sc order, which respects program order.

    Lastly we have the case $a \in e \tau_2\downarrow_\sigma$ and $b \in \tau_1\downarrow_\sigma e$.
    There must be at least one event $e_k$ such that $e_k \in e \tau_2\downarrow_\sigma$ and $e_{k+1} \in \tau_1\downarrow_\sigma e$, that is, $e_k \to[\rel_k''] e_{k+1}$ is a backward edge.
    Since \oppo-, \rf-, $([\nlW];\pf)$- and \nfo-edges are forward edges, this backward edge must be one of $[\nrW];\ippo'';[\nF]$, $[\nlW];\ippo'';[\nF]$ or $[\nrW];\pf''$.
    In particular, $e_k \to[\ib] e_{k+1}$, since $\ippo'' = \ippo$ and $\pf'' = \pf$.

    Let $\ell \leq k$ be such that $e_\ell \to[\ib] e_k$ but $(e_{\ell-1}, e_k) \not\in\ib$.
    Such an $\ell$ must exist since $a = e_0 \to[\ib] e_k \to[\ib] e_{k+1}$ would be a backward $([\Inst];\ib)$-edge with respect to $\tau$, which cannot exist.
    Furthermore, none of $e_\ell, e_{\ell+1}, \dots, e_{k-1}, e_k$ can be an instantaneous event, since this would again yield a backward $([\Inst];\ib)$-edge.
    In particular, all of them, including $e_\ell$, are write events.
    
    Now, consider $(e_{\ell-1}, e_\ell) \in \ib''\setminus\ib$.
    Since $\rel'' = \rel$ for $\rel \in \set{\ippo, \rf, \pf}$, it must be the case that $e_{\ell-1} \to[\nfo''] e_\ell$ and $e_{\ell} \to[\nfo] e_{\ell-1}$.
    $e_\ell \in \nlW$ is not possible since it implies $e_{\ell-1} \in \nlR$, but $[\nlR];\sqp;[\nlW] \subset \ib$, which contradicts $e_{\ell} \to[\nfo] e_{\ell-1}$.
    Thus, $e_\ell \in \nrW$ and $e_{\ell-1} \in \nrR$.
    We conclude that the writes $e_\ell, \dots, e_k$ cannot be CPU writes, since there are no \ib-edges from NIC writes to CPU writes.
    In particular, all of these writes and also the NIC fence / poll $e_{k+1}$ must be on the same queue pair.

    Now, consider the matching NIC local write $w$ to the NIC remote read $e_{\ell-1}$.
    We have two cases depending on whether the backward edge $e_k \to e_{k+1}$ is due to a memory fence or a poll.
    In the former case, we obtain $e_{\ell-1} \to[\ippo] w \to[\ippo] e_{k+1}$ which is a backward $([\Inst];\ib)$-edge.
    In the latter case, since $w \to[\po] e_k$ and $e_k \to[\pf] e_{k+1}$, there must be a poll event $p$ such that $w \to[\pf] p$ and $p \to[\po] e_{k+1}$.
    We obtain the backward $([\Inst];\ib)$-edge $e_{\ell-1} \to[\ippo] w \to[\pf] p \to[\ippo] e_{k+1}$.
    This is a contradiction since \ob cannot have any backward edges.
\end{proof}

\begin{corollary}
\label{cor:T_prime_prime_violating}
    The traces $T$ and $T''$ have the same SC relation: $T.\sc = T''.\sc$.
\end{corollary}
\begin{proof}
    This follows immediately from \autoref{lem:rel_equals_rel_prime_prime}.
\end{proof}
\newpage
\section{Instrumentation Details}

\subsection{Detailed RDMA-Consistency Checks}\label{sec:rdma_consistency_details}

This section provides the complete case analysis for ensuring RDMA-consistency during the instrumentation process.
To ensure that an event can be added to the trace in an RDMA-consistent way, the program must check that each \ob-edge is a forward edge.
Additionally, it must verify that \rf is well-formed, i.e., that the value of a read event matches the value of its write event.
We check each relation of \ob as follows:
\begin{itemize}
    \item \oppo:
    This is the most difficult case since there are many different subcases and it can be a backward edge.
    Since the current event is \po-later than all previously simulated events, it can always be added to $\tau_2$ without creating a backward \oppo-edge.
    However, it can only be added to $\tau_1$ if there is not already an event in $\tau_2$ that would create a backward edge.
    
    During the simulation, for each type of event as in \autoref{fig:ippo-oppo}, we track whether an event of that type has already been added to $\tau_2$.
    This tracking is done per queue pair.
    When the program decides to add an event to $\tau_1$, it checks that there is no \oppo-earlier event in $\tau_2$.
    For example, a NIC remote write $\nrW$ can only be added to $\tau_1$ if there is neither a CPU event, a NIC local read $\nlR$, a NIC remote write $\nrW$, nor a remote fence $\nF$ in $\tau_2$ with respect to the same queue pair.
    \item \rf:
    All \rf-edges are forward edges.
    The program only needs to ensure that the value of the read event matches the current value in the memory.
    \item $([\nlW];\pf)$:
    The program keeps a counter of how many non-polled NIC writes have been executed.
    A poll can execute only if the counter is non-zero, in which case it decrements the counter.
    A poll would become a backward edge if placed in $\tau_1$ while polling from a NIC local write from $\tau_2$.
    We prevent this by tracking the first such event that forces its poll to go into $\tau_2$.
    Note that once one poll is in $\tau_2$, all \po-later events must also be in $\tau_2$, since polls are CPU events.
    \item \nfo, \rb, \mo:
    All edges of these types are always forward edges.
    \item $([\Inst];\ib)$:
    This is another difficult case with many subcases to check.
    Consider an $([\Inst];\ib)$ backward edge, i.e., a path $e_0 \to[\rel_1] e_1 \to[\rel_2] \dots \to[\rel_n] e_n$ such that $\rel_i \in \set{ \ippo, \rf, \pf, \nfo }$ for all $i = 1, \dots, n$ and $e_n <_\tau e_0$.
    Consider a path of shortest length.
    The last edge in it must be a backward edge.
    Since we have already ensured that \oppo, \rf, $([\nlW];\pf)$, and \nfo have no backward edges, it must be either $\nW \to[\ippo] \nF$ or $\nrW \to[\pf] \poll$.
    
    \underline{Case $\nW \to[\ippo] \nF$:}
    Since \nW is not an instantaneous event, we distinguish various cases of $e_{n-2} \to[\rel_{n-1}] e_{n-1}$.
    If $\rel_{n-1} = \ippo$ and $e_{n-2}$ is an instantaneous event, then $e_{n-2} \to[\oppo] e_n$, since \po-edges from any instantaneous event to any type of NIC write on the same queue pair are in \oppo.
    Thus, this case cannot occur since the program already checks \oppo.
    The case where $\rel_{n-1} = \ippo$ and $e_{n-2}$ is a NIC write leads to a shorter \ib-path $e_0 \to[\rel_1] \dots \to[\rel_{n-2}] e_{n-2} \to[\ippo] e_n$, which is a contradiction to the minimality of $e_0, \dots, e_n$.
    Since $e_{n-1}$ is a write, $\rel_{n-1}$ can be neither \rf nor \pf.
    Thus, the last case is that $e_{n-1}$ is a NIC read and $\rel_{n-1} = \nfo$.
    We can eliminate $e_{n-2} \to[\po] e_{n-1}$, since this would imply $e_{n-2} \to[\oppo] e_n$ again.
    Thus, $e_{n-1} \to[\po] e_{n-2}$, but since $e_{n-2} \to[\nfo] e_{n-1}$ is in \ob, this cannot also hold for the former.
    The only pair of earlier write and later read that is not in \ob is a local write followed by a remote read.
    We now analyse the position of $e_{n-2}$ within $\tau$.
    Since $e_{n-1} \to[\po] e_n$ and events within $\tau_1$ and $\tau_2$ respect \po, respectively, it must be the case that $e_{n-1} \in \tau_2$ and $e_n \in \tau_1$.
    Similarly, since $e_{n-2} \to[\nfo] e_{n-1}$ but $e_{n-1} \to[\po] e_{n-2}$, it follows that $e_{n-2}$ cannot also be in $\tau_2$.
    Lastly, if $e_{n-2} <_\tau e_n$, then $e_0 \to[\ib] e_{n-2}$ would be a shorter $([\Inst];\ib)$-path forming a backward edge.
    Thus, $e_n <_\tau e_{n-2}$, which means that $e_{n-2}$ is simulated after $e_n$.
    All the program needs to check is that it does not add a NIC local read to $\tau_1$ while there is already a NIC local write in $\tau_2$ which has a \po-later remote fence on the same queue pair that is in $\tau_1$.
    The latter condition can be easily checked and tracked whenever a fence is added to $\tau_1$, since we already track whether there is a NIC local write in $\tau_2$.

    \underline{Case $\nrW \to[\pf] \poll$:}
    We proceed similarly and distinguish cases based on $e_{n-2} \to[\rel_{n-1}] e_{n-1}$.
    If $e_{n-2}$ is a CPU event and $\rel_{n-1} = \ippo$, then we already have $e_{n-2} \to[\oppo] e_n$.
    
    If the relation is \ippo but $e_{n-2}$ is a NIC local read or a NIC fence, we analyse the position of $e_{n-2}$ within $\tau$.
    Suppose $e_{n-2} <_\tau e_n$; then $e_0 \to[\ib] e_{n-2}$ is a shorter path.
    Thus, $e_n <_\tau e_{n-2}$.
    But since $e_{n-2} \to[\po] e_n$, it follows that $e_{n-2}$ cannot be in the same part of $\tau$ as $e_n$.
    Therefore, $e_{n-2} \in \tau_2$.
    The program prevents this case whenever it adds a NIC remote write to $\tau_2$.
    If this write is preceded by a NIC local read or a NIC fence in $\tau_2$, then it is marked as forcing its poll to be in $\tau_2$, similar to how any NIC local write would.
    
    In the next case, $e_{n-2}$ is another remote write and the relation is still \ippo.
    If $e_{n-2} <_\tau e_n$, we consider $e_0 \to[\ib] e_{n-2}$ instead, which is shorter.
    Otherwise, since $e_{n-2}$ and $e_{n-1}$ are on the same queue pair and $e_{n-1}$ is polled by $e_n$, there must be some poll event $e'$ that is polling from $e_{n-2}$.
    Since $e' <_\tau e_n <_\tau e_{n-2}$, it follows that $e_0 \to[\ib] e_{n-2} \to[\pf] e'$ is a shorter backward path, which contradicts minimality.
    
    A similar case is where $e_{n-2}$ is a NIC remote read and the relation is \nfo.
    Here, we observe that the NIC remote read must have a matching NIC local write $e''$ that was polled by some poll $e'$ before $e_n$.
    We obtain $e' <_\tau e_n <_\tau e_{n-2} <_\tau e''$, which implies a backward \ob-edge from the NIC local write $e''$ to the poll $e'$.
    This observation excludes this case and finishes the case distinction, since the NIC remote write $e_{n-1}$ cannot be the target of any \rf- or \pf-edge.
\end{itemize}

\subsection{Detailed SC-Cycle Detection}\label{sec:sc_cycle_detection_details}

This section provides the complete details of how the instrumented program detects SC-cycles.

To detect SC-cycles, we non-deterministically guess the entry and exit points of the SC-cycle in each thread.
The entry and exit points are the \po-first and \po-last events of the thread on the cycle, respectively.
Clearly, the entry and exit points of the same thread are related by \po.
The exit point of one thread and the entry point of the next thread must be related by \rf, \mo, or \rb.
All three of these relate a pair of read or write events on the same variable, where at least one (but possibly both) of the events is a write event.

The implementation can be simplified through the observation that any pair of such events must be \sc-related, either directly through one of these relations or indirectly through a chain of them.
For example, an earlier write $w$ is related to a later read $r$ by either $w \to[\rf] r$ or a chain $w \to[\mo] w' \to[\mo] \dots \to[\rf] r$.

To ensure that the exit point of one thread is later in $\tau$ than the entry point of the next thread in the cycle, we distinguish two cases:
If the exit point was simulated before the entry point, then it is sufficient that either the exit point is in $\tau_1$ or the entry point is in $\tau_2$.
On the other hand, if they were simulated in the opposite order, then both conditions must hold.
This is checked during the simulation of the second event (the one simulated later), and a successful matching pair of events is signaled to the main thread via a flag variable.
After all threads have finished their simulation, the main thread checks whether all pairs of entry and exit points on the guessed cycle were successful.

\end{document}